\documentclass[aps,prl,notitlepage,superscriptaddress,reprint]{revtex4-1}

\usepackage[utf8]{inputenc} %for PDFLatex
\usepackage{siunitx}
\DeclareSIUnit{\sample}{Sa}

\usepackage{graphicx}

\usepackage{amsmath} 
\usepackage{placeins} 

\usepackage[dvipsnames,svgnames,x11names]{xcolor}
\usepackage{float}
\usepackage{textcomp}
\usepackage{siunitx}
\usepackage{scrlayer-scrpage}

\usepackage{amssymb}
\usepackage{amsthm} %math proofs

\usepackage{xr-hyper}

\usepackage{pgfplots} %for plotting the illustration of the timeline and bimodal noise
\pgfplotsset{compat=newest}

\usepackage{hyperref}
\hypersetup{pdfprintscaling=None}

\ifdefined\drafting
	\usepackage[a4paper,left=0.25in,right=1.8in,head=21.0pt]{geometry}
\else
	
\fi

\ifdefined\showcomments{}
	\usepackage[commentmarkup=todo,todonotes={textsize=scriptsize,textwidth=1.6in}]{changes}
	\makeatletter \def\@captype{figure} \makeatother %https://tex.stackexchange.com/a/544899
\else
	\usepackage[final]{changes}
\fi

\newcommand{\markhigh}[1]{\bgroup\markoverwith				%necessary for \highlight to break line
  {\textcolor{#1}{\rule[-.5ex]{2pt}{2.5ex}}}\ULon}

\definechangesauthor[name=David Reifert, color=BlueViolet]{DR}
\definechangesauthor[name=Niels Ubbelohde, color=red]{NU}
\definechangesauthor[name=Review, color=Teal]{R}
\definechangesauthor[name=Martins Kokainis, color=ForestGreen]{MK}
\definechangesauthor[name=Slava Kashcheyevs, color=blue]{VK}

\ifdefined\removefig
	\usepackage{comment} 	
	\excludecomment{figure}	
		
\else
\fi

\newcommand\lr[1]{\left( #1 \right)}
\newcommand\lrv[1]{\left|  #1 \right|}
\newcommand\lrk[1]{\left[    #1 \right]}
\newcommand\lrb[1]{\left\lbrace #1 \right\rbrace}
\newcommand\lra[1]{\left\langle #1 \right\rangle}

\newcommand{\mbb}[1]{\mathbb{#1}}
\newcommand{\mb}[1]{\mathbf{#1}}
\newcommand{\probQ}{P_{-}}
\newcommand{\probO}{P_{0}}
\newcommand{\probP}{P_{+}}
\newcommand{\probPQ}{\ensuremath{P_{\pm}}}
\newcommand{\probQP}{\ensuremath{P_{\mp}}}
\newcommand{\Nsim}{N_{\mathrm{sim}}}

\DeclareMathOperator{\Dir}{Dir}
\DeclareMathOperator{\Expect}{E}
\DeclareMathOperator{\Var}{Var}
\DeclareMathOperator{\Cov}{Cov}
\newcommand{\imagI}{\mathrm{i}}
\newtheorem{claim}{Claim}
\newcommand{\numlistDef}[1]{\numlist[list-pair-separator ={, },list-final-separator ={, },round-mode = places,round-precision = 2]{#1}}
\newcommand{\barX}{\overline{X}}
\DeclareMathOperator{\IQR}{IQR}

\usepackage{xpatch}
\makeatletter
\patchcmd{\@ssect@ltx}
    {\addcontentsline{toc}{#1}{\protect\numberline{}#8}}
    {}
    {}
    {}
\makeatother

\begin{document}

\title{A random-walk benchmark for single-electron circuits}

\author{David Reifert}
% no longer valid? \email{david.reifert@ptb.de}
\affiliation{Physikalisch-Technische Bundesanstalt, 38116 Braunschweig, Germany}

\author{Martins Kokainis}
\affiliation{Faculty of Computing, University of Latvia, 19 Raina boulevard, LV-1586 Riga, Latvia}
\affiliation{Department of Physics,  University of Latvia, 3 Jelgavas street, LV-1004 Riga, Latvia}

\author{Andris Ambainis}
\affiliation{Faculty of Computing, University of Latvia, 19 Raina boulevard, LV-1586 Riga, Latvia}

\author{Vyacheslavs Kashcheyevs}
\affiliation{Department of Physics,  University of Latvia, 3 Jelgavas street, LV-1004 Riga, Latvia}

\author{Niels Ubbelohde}
\email[To whom correspondence should be addressed: ]{niels.ubbelohde@ptb.de}
\affiliation{Physikalisch-Technische Bundesanstalt, 38116 Braunschweig, Germany}

%\date{\today}
\begin{abstract}
Mesoscopic integrated circuits aim for precise control over elementary quantum systems. However, as fidelities  improve, the increasingly rare errors and component crosstalk pose a challenge for validating  error models and quantifying accuracy of circuit performance. Here we propose and implement a circuit-level benchmark that models fidelity as a random walk of an error syndrome, detected by an accumulating probe.  Additionally,  contributions of correlated noise, induced environmentally  or  by  memory, are revealed as limits of achievable fidelity by  statistical consistency analysis of the full distribution of error counts.			
Applying this methodology  to a high-fidelity implementation of on-demand transfer of electrons in  quantum dots we are able to utilize the high precision of charge counting to robustly estimate the error rate of the full circuit and its variability due to noise in the environment. As the clock frequency of the circuit is increased, the random walk reveals a memory effect.
This  benchmark contributes towards a rigorous metrology of quantum circuits.
\end{abstract}

\maketitle

Precise manipulation of individual quantum particles in complex single-electron circuits for sensors, quantum  metrology, and quantum information transfer \cite{Baeuerle2018,Pekola2013}
requires tools to certify fidelity and establish a scalable  error model.
A similar challenge arises in the gate-based approach to universal quantum computation \cite{DiVincenzo2000,Chow2012,Barends2014,Benhelm2008,Gaebler2016,Ballance2016} where benchmarking gate sequences  \cite{Emerson2005,Knill2008,Magesan2011,Huang2019,Erhard2019} are employed to validate independent-error models \cite{Arute2019} which are crucial for scaling towards fault-tolerance \cite{Aharonov2008,Fowler2012}.
Here, we introduce the idea of benchmarking by error accumulation to 
integrated single-electron circuits.
We experimentally realize  clock-controlled transfer of electrons through a chain of quantum dots, and describe the statistics of accumulated charge by a random-walk model. 
High-fidelity components and unprecedented accuracy of charge counting enable the detection of excess noise beyond the sampling error, the identification of the timescale for consecutive step interaction, and an accurate estimate for the failure probabilities of the elementary charge transfer. Abstracting errors from component to circuit level opens a path to
leverage  charge counting for microscopic certification of electrical quantities challenging the precision of metrological measurements \cite{Stein2015}, and to introduce fidelity control in building blocks of quantum circuits \cite{Takada2019,Mills2019, Nakajima2019, Freise2020}.

In quantum metrology, stability and reproducibility  of the environment  for elementary quantum entities (photons, qubits, electrons) and their  uncontrolled interactions set the practical limits on the precision of quantum circuits \cite{Smirne2016}, which approach the fundamental quantum limits, i.e.\  counting shot noise for independent identical particles, or the Heisenberg limit for entanglement-enhanced measurements \cite{Giovannetti2011}. In particular,  accurate benchmarking of fidelity in the presence of long-term drifts and memory  is difficult but essential for the validation of the precision of quantum standards.
Identifying and quantifying the residual error, i.e.\ any deviation from the perfect performance of a circuit, define the challenge to be answered by the random-walk benchmarking for high-precision single-electron current sources. Validating  consistency of the error model by  statistical testing ensures the robustness of the fidelity estimates, which is an actively studied problem in the  related context of assessing quantum computation platforms \cite{Ball2016,Epstein2014,OMalley2015,Arute2019}.

The random-walk benchmarking addresses the question of uniformity in time of repeated identical operations by error accumulation. The error signal (syndrome) considered here is the discrete charge stored in the circuit after executing a sequence of $t$ operations. The measured
 deviation $x$ in the number of trapped electrons  is modelled by the probability $p_x^t$  for a random walker to reach integer coordinate $x$ from initial position of $x=0$ in $t$ steps, see Figure \ref{fig:fig1}. 
In the desired high-fidelity limit of near-deterministic on-demand transfer of a fixed number of electrons
any residual randomly occurring errors that alter $x$ will be very rare 
and the walker will remain stationary most of the time, with occasional steps of length one. Here we study to what extent two single-step, $x \to x \pm 1$, probabilities  $P_{\pm}$ describe  the statistics of $x$ collected by repeated operation of the  circuit, and how deviations from independent error accumulation can be detected and quantified, revealing otherwise hidden physics.
The baseline random-walk model with  $t$- and $x$-independent $P_{\pm}$ predicts the following distribution:
 \begin{multline} 
	p^t_{x \ge 0}  =  (1-\probP-\probQ)^{t-x} (\probP)^x \binom{t}{x}  \times \\
	{}_2F_1\left(\frac{x-t}{2},\frac{x-t+1}{2};x+1;\frac{4\probP\probQ}{(1-\probP-\probQ)^2} \right)
	\label{eq:px}
\end{multline}
with $p^t_{x < 0}$ obtained from Eq.~\eqref{eq:px} by $x \rightarrow -x$ and $\probPQ \rightarrow \probQP$ (see derivation in Supplementary Note \ref{sec:basic_model}).
Here the first term of the product describes decay of fidelity that is exponential in $t$, while the binomial coefficient and the Gaussian hypergeometric function ${}_2F_1$ (here a polynomial of order at most $t$) take into account the self-intersecting paths as single-step errors accumulate and partially cancel at large $t$ (see Figure~\ref{fig:fig1}a).

\begin{figure*}[htbp]
	\centering
	\includegraphics[width=510pt]{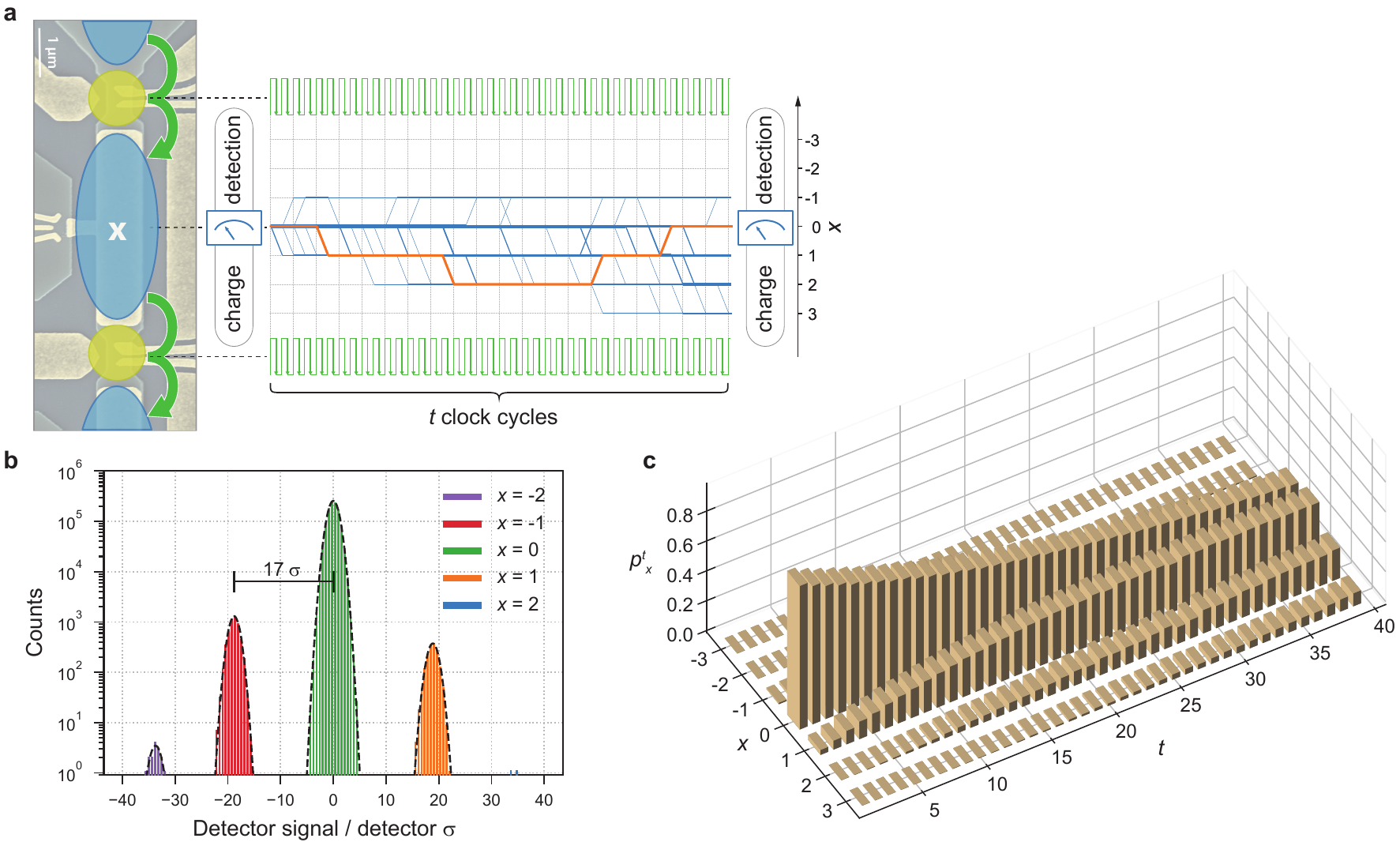}
	\caption[Random Walks]{(a) Sample micrograph and measurement scheme. After the initial charge measurement $t$ clock cycles are applied. The paths taken by 30 simulated walkers (using error rates extracted from the counting statistics) are represented by blue lines, transitioning every clock cycle in $x$ by a step of $-1,0,+1$. The frequency with which each branch is visited is indicated by the linewidth. A final charge measurement yields the end-point of the random walk as the difference between initial and final charge. The orange line exemplifies a single random walk with self-intersections. (b) Signal to noise ratio: a (typical) histogram of the differential charge detection signal with the identified difference in electron number indicated by color. The peak separation is shown in units of the  Gaussian noise amplitude $\sigma$ (black dashed lines indicate the corresponding Gaussian fits). (c) Measured  statistics of finding the walker at position $x$ after $t$ steps.} 
	\label{fig:fig1}
	
\end{figure*}

Experimentally, the high-fidelity circuit for electron transfer is realized by a chain of quantum dots in which the first and the last dot are operated as single-electron pumps \cite{Kaestner2015} and the central dot provides the error signal as shown in Figure \ref{fig:fig1}.
A clock of frequency ($f=$ \SIrange{30}{300}{\mega\hertz}) drives the pumps to transfer one electron per cycle through the chain (from top to bottom in Figure~\ref{fig:fig1}).

Within one clock-cycle, the entrance barrier to the dynamic quantum dot is lowered and raised by the pump stimulus, isolating one electron from the source reservoir and then ejecting it over the high exit barrier; barrier height asymmetry between entrance and exit defines the transfer direction \cite{Kaestner2008}. The operating points of the  pumps are chosen to minimize and approximately balance the error probabilities of transferring either zero or two electrons instead of one  (with a slight bias towards zero-electron transfers, as this error rate only increases exponentially and not double-exponentially with deviations from the optimal operating point \cite{Kashcheyevs2010}). The working points of the pumps are not retuned when operating the full circuit. Reproducible formation of quantum dots \cite{Gerster2018} allows to demonstrate high-fidelity operation of the circuit event at zero magnetic field, at which readout precision is enhanced by cryogenic reflectometry.

The excess charge $x$ from accumulating errors is inferred from a differential measurement by a charge detector capacitively coupled to the central dot, reading out the detector state  before and after each sequence transferring $t$ electrons. As tunneling events are only enabled by the clocked stimulus applied to the pumps, a long detector integration-time up to \SI{1}{ms} can be chosen for unambiguous identification of $x$ with a signal to noise ratio of 17  (Figure 1b).
A full histogram of detector states before and after the transfer sequence allows to reconstruct the shape of the Coulomb blockade peak resonance utilized by the charge detector and provides rigorous classification thresholds for the identification of $x$. The sequence of electron transfer and charge detection  is repeated with {the} repetition rate limited by the detector integration-time (up to \SI{4}{\kilo\hertz}), until a set number of counts ($N=$\numrange{1e5}{2e6})  is accumulated. Any deviations not aligned with the measurement timing, such as instabilities in the charge detector, are readily recognized and discarded, while unintended charge transitions during the operation of the pump are counted and correctly identified as errors. 

Although the individual accuracy of the active components  can exceed metrological precision~\cite{Giblin2019}, their simultaneous operation in a mesoscopic circuit \cite{Fricke2011} precludes the prediction of transfer fidelity from  component-wise characterization due to interactions and cross-talk between the elements in the chain, exemplifying the need  for circuit-level benchmarking. Experimental evidence for strong discord between component-wise and circuit-level characterization is  given in
Supplementary Note  \ref{sec:compare_pumps}.

{Here} we report the measurement results on two devices: device A introduces the methodology to resolve effects beyond statistical noise of independent error accumulation in a high-fidelity circuit, while device B demonstrates the  effects of memory with increased repetition frequency. Both devices share very similar device geometries and parameters.

Figure \ref{fig:fig2}a shows the counting statistics measured for device A  at $f=\SI{30}{MHz}$ for  $t$ up to $10^4$ compared to predictions of the  baseline model.  General trends expected from the random walk are evident:
 for short sequences, $t < 1000 \ll (\probP \, \probQ)^{-1/2}$, 
the power-law rise of the  probabilities  $p^{t}_{|x|>0}$ corresponds to the exponential decay of error-free transfer fidelity $p^{t}_0$ which remains close to $1$. For longer sequences the distribution spreads and the weight of self-intersecting  paths (e.g. orange line in Fig.\,1) increases, in accordance with Eq.~\eqref{eq:px}.

\begin{figure*}[htbp]
	\centering
	\includegraphics[width=510pt]{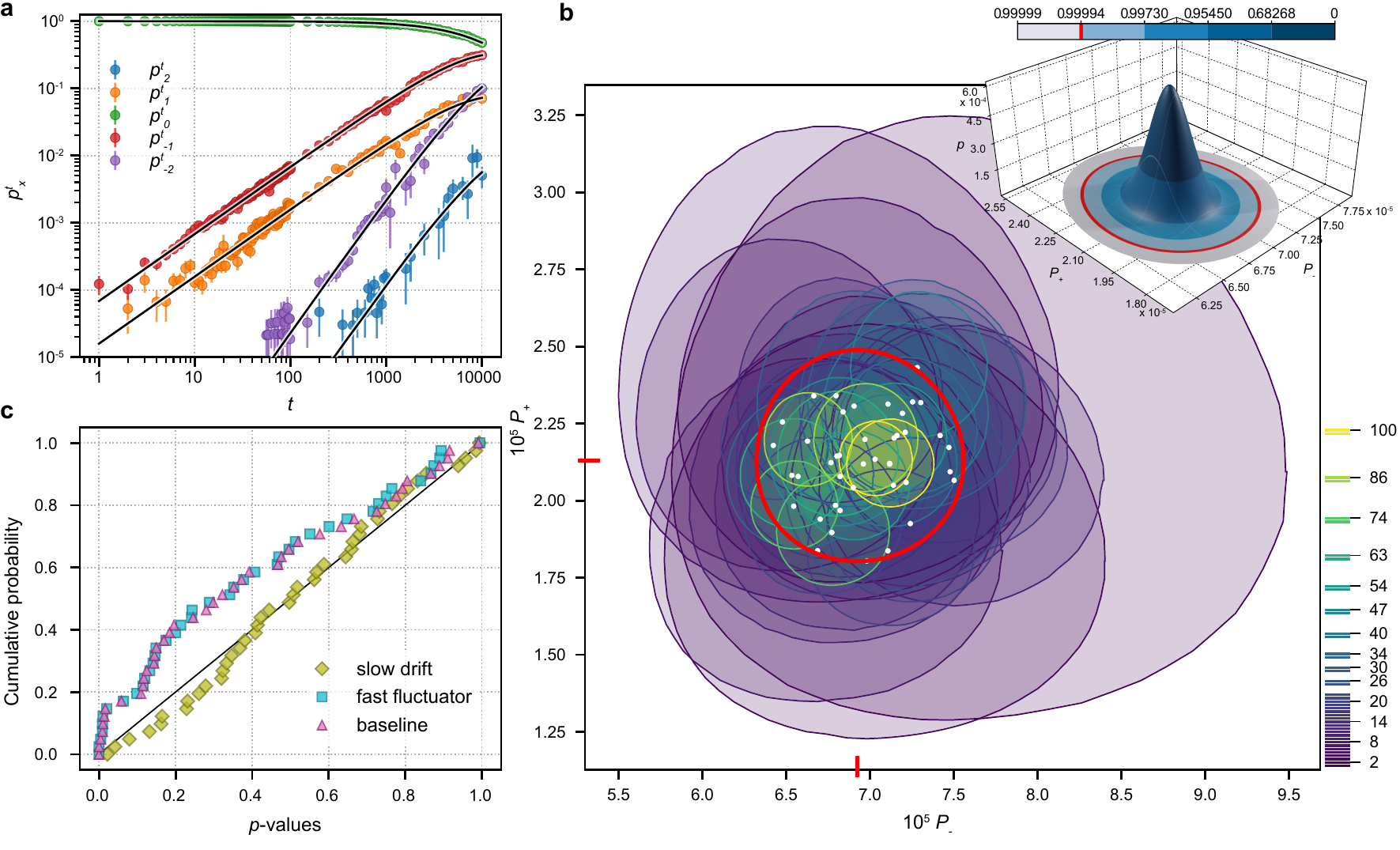}
	\caption[Consistency regions]{(a) Measured $p^t_x$ for device A; error bars are given by the standard deviation of the binomial distribution, solid lines show a least-squares fit of Eq.~\eqref{eq:px}. (b) Likelihood-maximising $\probPQ$ (white dots) and $p>0.05$ consistency regions estimated separately for each sequence length (coded by color). The inset shows the probability density function of the Dirichlet distribution with parameter $\boldsymbol\alpha=(\numlist[list-final-separator = {, },round-mode = places,round-precision = 2]{2.427195e3;3.504960e7;7.468530e2})$. The corresponding global best-fit values for $P_\pm$ are marked by red lines on the axes of the consistency-region plot. The color scale indicates the level of confidence at different coverage factors $k$ for a symmetric normal distribution; the red circle in both plots and the marker in the color scale indicates the region corresponding to $k=4$. (c) Empirical cumulative distribution of  $ p $-values for different models in comparison to the uniform distribution (black line).}
	\label{fig:fig2}
\end{figure*}

The key question for random-walk benchmarking is whether the uncorrelated residual randomness defined by two probabilities $P_{+}$ and $P_{-}$ predicts the entire probability distribution. This question is answered in three steps:
(i)  significance testing of deviations from the baseline model as a statistical null hypothesis to delineate the inevitable sampling error from model error; (ii)
 extending the  model to accommodate correlated excess noise \cite{Mavadia2018} detected in the first step;
 (iii) perform parameter estimation of the noise model
 that yields  average values of $P_{\pm}$ with an estimate of the variability.

For consistency testing, we have increased the number $N$ of samples per sequence by a factor of $\sim 10$, and limited $t$ to 100.  Fisherian significance tests \cite{Christensen2005} are used to  define  consistency regions of $p$-value greater than $0.05$ in the parameter space $(P_{+}, P_{-})$ where the baseline model cannot be rejected at this significance level (see Methods).  Figure \ref{fig:fig2}b shows quasielliptic consistency regions  
computed for each sequence length $t$
separately, randomly clustering in a tight area with the sizes shrinking roughly as $\sim 1/\sqrt{t}$, as expected. Their overlap is only partial:
best-fit global $(P_{+},P_{-})$ estimated from maximal likelihood (marked on the axes of Figure \ref{fig:fig2}b) lies outside of $7$ regions out of $42$. A more rigorous test on whether this inconsistency can be explained by sampling error alone is provided by 
Fisher's meta-analysis method (Figure \ref{fig:fig2}c): under the null-hypothesis,  the cumulative distribution of $p$-values obtained separately for each sequence length $t$ should be uniform (a straight line)~\cite{Borenstein2009,Fisher1932} (see Supplementary Note \ref{sec:comb_p-val}) which is not the case for the best-fit baseline model  (triangles in Figure \ref{fig:fig2}c). Quantitatively, the  baseline model  yields global Fisher's combined $p < 3 \times 10^{-6}$, and hence is statistically rejected. We attribute this incompatibility to excess noise
due to imperfections in the physical realization of the baseline model.
Nevertheless, the partial overlap and the tight clustering observed in Figure~\ref{fig:fig2}b suggests that the excess noise is rather small.
We model the excess noise as stochastic variability of $P_{\pm}$, and check whether it can be plausibly explained by the presence of two-level fluctuators~\cite{Paladino2014}.

To quantify the excess noise, the model is now extended (part (ii) of the outline above) by drawing the step probabilities $P_{\pm}$ randomly from a Dirichlet distribution \cite{Ng2011,Johnson97} (Supplementary Note \ref{sec:dirichlet}) over the standard 2-simplex; the corresponding parameters $\boldsymbol\alpha=\lrb{ \alpha \, \lra{P_-} , 
\alpha \,( 1-\lra{P_+}- \lra{P_{-}}) , \alpha \, \lra{P_+} }$  are specified by two means,  $\langle P_{\pm}\rangle$, 
and one additional concentration parameter $\alpha$ which controls the variance, $\Delta P_{\pm}^2 =\lra{P_{\pm}}(1-\lra{P_{\pm}})/(\alpha+1)$.
The Dirichlet distribution is strongly peaked near the mean point for $\alpha \gg \lra{P_{\pm}}^{-1}$, and always guarantees $0 \le P_{\pm} \le 1$.
% = \lra{ ( P_{\pm} -\lra{P_{\pm}})^2}=
This extra randomness can be introduced at different timescales~\cite{Ball2016}. Uncorrelated noise (new $P_{\pm}$ after each step of a walk)  is equivalent to the baseline model with $P_{\pm} \to \lra{P_{\pm}}$, and is already ruled out by the significance tests above.
 We compare a ``fast fluctuator'' model in which a new pair of $P_{\pm}$ is drawn independently after completion of each individual random walk versus a ``slow drift'' model  in which the values of $P_{\pm}$ are randomly reset only after all $N$ realizations for a fixed number of steps have been collected (precise excess noise model definitions are given in Supplementary Note \ref{sec:model-1} and \ref{sec:model-2}, and the data acquisition timeline is illustrated in Fig.~\ref{fig:TLF}). Although short of proper time-resolved noise metrology \cite{OMalley2015}, contrasting these two correlated-noise models gives an indication of the relevant timescales (nanoseconds versus half-hour in the experiments). 
The sensitivity of Fisher’s significance testing makes it possible to distinguish between the two models, which cannot be resolved by the second moment of  $\langle p_x^t \rangle$   as utilized, e.g., for noise-averaged fidelities in randomized benchmarking of quantum gates  \cite{Mavadia2018}.
The results of Fisher's  combined test (see Figure \ref{fig:fig2}c) favour the ``slow drift'' ($p=0.71$) over the ``fast fluctuator'' ($p<3 \times 10^{-6}$) model. The corresponding best-fitting Dirichlet distribution (parameters indicated by red lines on the axes of Figure \ref{fig:fig2}b and plotted in the inset)  gives $1\, \sigma$ uncertainty estimates $P_{-}=(6.92 \pm 0.14) \times 10^{-5}$
and $P_{+}=(2.13 \pm 0.08)\times 10^{-5}$.
Parametric instability at only a few-percent level validates a suitably extended random walk model as a robust representation of error accumulation in  this high-fidelity single-electron circuit.

In order to gain insight into a possible physics mechanism for excess noise and illustrate the robustness of statistical methods, we have simulated the experimental timeline using a random walk model with $P_{\pm}$ parameters subjected to $1/f$ noise from an ensemble of independent two-level fluctuators (see Supplementary Note \ref{sec:noise-sim}). The results follow the  general pattern outlined above:
(i) for a fixed size of the statistical sample, there is a threshold in the excess noise amplitude above which the data contradict both the baseline and the fast-fluctuator models but remain consistent with the slow-drift model.  This threshold corresponds to excess noise sufficiently  affecting probabilities of multiple errors per burst  to reveal inconsistency with Eq.~\eqref{eq:px} in the tails ($|x|>1$) of the error syndrome distribution $p_x^t$. (ii) The estimated best-fit $\Delta P_{\pm}$ parameters correlate well with the standard deviation of the $P_{\pm}$ in the underlying  simulation. (iii) Even a single fluctuator with a fixed switching rate (bimodal distribution of $P_{\pm}$  and a Poisson distribution of switching times \cite{Jenei2019}) can generate detectable excess noise still consistent with our Dirichlet-based statistical models.
As for the physics of the real device in a noisy environment, the simulations favor an explanation of the detected excess noise by the presence of multiple charge fluctuators over a single two-level system due to the absence of a bimodal signature in Fig.~\ref{fig:fig2}b.
In conclusion, accurate statistics of error counts can give enough sensitivity to
reliably estimate the baseline error rates $P_{\pm}$ and even  capture a fingerprint of long-time correlations
in the environment.

The methodology to quantify independent error accumulation described above  makes it possible to probe the effect of increased clock frequency on the circuit and thereby investigate response times of the electron shuttle and interactions between subsequent steps.
In device B the error rates are $P_{-}=(6.31 \pm 0.23) \times 10^{-3}$ and $P_{+}=(2.71 \pm 0.043) \times 10^{-2}$ at the same frequency of \SI{30}{MHz} as  device A investigated above. Ten-fold increase of the clock frequency to \SI{300}{MHz} is introduced by uniform time compression of signals controlling the transfer  operations;  
the resulting counting statistics is presented in Figure \ref{fig:fig3}a (circles). The random-walk model with constant $P_{\pm}$, described by Eq.~\eqref{eq:px}, no longer applies even qualitatively, which raises the question whether the fidelity of the circuit has decreased to a point where errors can no longer be considered rare as outlined in the beginning. This question is answered in the negative with the help of the following {theorem} defining a spread condition, which sets a precise bound on the applicability of the random-walk approach with  possibly non-stationary error rates: If distributions $(p_x^t)$ and $(p_x^{t+1})$ satisfy
\begin{align}
\sum\limits_{y=-\infty}^{x-1}p_y^t \leq \sum\limits_{y=-\infty}^{x}p_y^{t+1} \leq \sum\limits_{y=-\infty}^{x+1}p_y^{t} \quad \text{for all } x,
\label{eq:sc}
\end{align}
then there exists a set of transition probabilities $P_{\pm 1}^{(x,t)}$ such that $(p_x^{t+1})$ is generated from $(p_x^t)$ by a Markov chain
$p_{x}^{t+1}=p_x^t + \sum_{s=\pm1} \left [  P_{s}^{(x-s,t)} p_{x-s}^t  - P_{s}^{(x,t)} p_{x}^t \right ]$.
Conversely, any discrete-space, discrete-time random walk with steps of lengths at most $1$ (our definition of a high-fidelity circuit) satisfies the spread condition \eqref{eq:sc}, see Supplementary Note \ref{sec:proof-spread-condition} for proof of both claims.

\begin{figure}[tb]
	\centering
	\includegraphics[width=246pt]{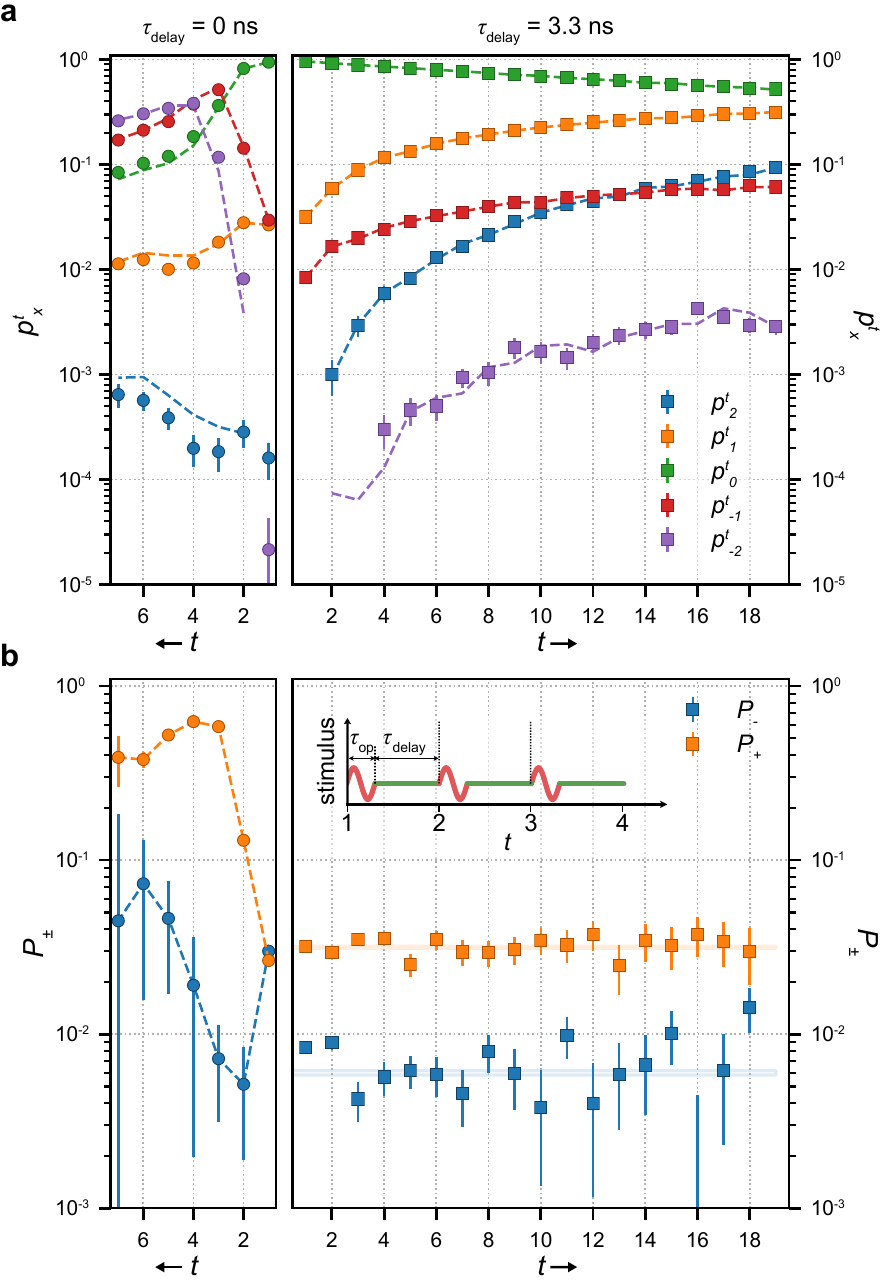}
	\caption[memory time]{(a) Measured $p^t_x$ for device B at a clock frequency of \SI{300}{\mega\hertz} and  $\tau_{\text{Delay}}=\SI{0}{\second}$ (left, $t$-axis inverted) and $\tau_{\text{Delay}}=\SI[parse-numbers = false]{3.\overline{3}}{\nano\second}$ (right). Dashed lines represent $p_x^t$ predicted by deconvolved single-step error rates and $p^{t-1}_x$. (b) Single-step error rates $\probPQ^t$ for $\tau_{\text{Delay}}=\SI{0}{\second}$ (left, $t$-axis inverted, dashed lines show guide to the eye) and $\tau_{\text{Delay}}=\SI[parse-numbers = false]{3.\overline{3}}{\nano\second}$ (right, translucent area corresponds to the $1\,\sigma$ uncertainty estimates). Inset depicts the timing diagram of the sequence -- a stimulus of duration $\tau_{\text{op}}$ drives the transfer operation followed by delay time $\tau_{\text{Delay}}$ before the next step.}
	\label{fig:fig3}
\end{figure}

We find that the distributions measured on device B do satisfy the spread condition \eqref{eq:sc} as long as all $x$ are fully resolved in counting ($t \leq 6$).
We estimate the  non-stationary but $x$-homogeneous single-step error probabilities of the corresponding Markov chains, $P_{\pm 1}^{(x,t)}=P_{\pm}^t$, by a numerical deconvolution of the Markov process equation (Supplementary Note \ref{sec:deconvolution-of--probpq-values}).
The resulting error rates $P_\pm^t$ in Fig. \ref{fig:fig3}b provide reasonable prediction (dashed lines) of the measured $p_x^t$ in Fig. \ref{fig:fig3}a (circles).
The $t$-dependence of $P_{\pm}^t$ is strong and reproduced well above the noise.
This implies memory: probabilities for the next step depend on how many steps have taken place before.
$P_{\pm}^t$ do not saturate within $t\le 6$ indicating a long memory time of more than $6\,\tau_{\text{op}} =\SI{20}{ns}$.

To probe this memory effect, we introduce a delay time  $\tau_\mathrm{Delay}$ between otherwise unaltered signals driving the transfer operations thus extending the physical time $f^{-1}$ corresponding to a single step of the random walk from $\tau_{\text{op}}$ to $\tau_{\text{op}} + \tau_{\text{Delay}}$ as sketched in Figure \ref{fig:fig3}b. With increasing delay, a gradual reduction of the $t$-dependence in $P_{\pm}^{t}$ is observed until, for $\tau_\mathrm{Delay} > \SI{3}{\nano\second}$ (see right part of Fig.\,\ref{fig:fig3}a and b), the stationary behaviour consistent with the baseline model is recovered. Surprisingly, $\tau_{\text{Delay}}$ sufficient to recover stationary behaviour is on the order of a single step duration $\tau_{\text{op}}$, significantly shorter than the number of steps with pronounced memory effect at $\tau_{\text{Delay}}=\SI{0}{\nano\second}$ (Figure \ref{fig:fig3}b). Both times are significantly longer than the expected timescales in GaAs systems for relaxation  via electron-electron or phonon interaction \cite{Ridley1991,Snoke1992,Molenkamp1992}, and raise the need for a dedicated investigation. In Fig.\,\ref{fig:fig3}b $P^t_{\pm}$, estimated at each $t$ by deconvolution (squares), are compared with the confidence intervals of the ``slow-drift'' model with stationary $P_{\pm}$ (color bands). The comparison shows good agreement and is consistent with our framework for random-walk benchmarking of high-fidelity single-electron circuits. 
For the showcased device, 
circuit-level interactions and  memory effects significantly lower the  attainable clock speed compared to record frequencies for individual pumps reported in the literature \cite{Yamahata2016}.
However, benchmarking by error accumulation introduces a  tool to investigate these limitations and  identify possible mitigation-techniques since $\tau_{\text{op}}$ and $\tau_{\text{Delay}}$ can be freely adjusted with error rates still accurately estimated on the circuit level, as long as these remain within the high-fidelity bound monitored by the spread condition.

In conclusion, the view of single-electron components as elements of a digital circuit has enabled an abstract and universal description of fidelity in terms of the random walk of an error syndrome. 
Accumulation of errors over long sequences allows to probe fast and accurate  operations beyond the bandwidth of a slow single-charge detector. The accompanying statistical methodology quantifies the stability of the error process and uncovers short memory times,
both of which are elusive to direct observation.
In quantum metrology, an accurate estimate of the circuit error has an immediate application: the variance of the current $I=(I_s + I_d)/2$ flowing into ($I_s$) and out of ($I_d$) the circuit is given by the variance of the differential charge $x$, which corresponds to the displacement current $I_s-I_d = e f x/t$. Hence, the variance of $x$,   $\Delta x^2 \approx (\langle P_{+}\rangle +\langle P_{-} \rangle) \, t + (\Delta P_{+}^2 +\Delta P_{-}^2) \, t^2$,   provides a bound for the deviation of the current $I$ from the error-free value $ e f$, enabling  counting-verification of a primary standard for the ampere.
In the broader context,
sensitive tests of single-electron circuits create new ground for developing benchmarking techniques of engineered quantum systems.

\section*{Methods}
\label{sec:methods}

Devices A and B were fabricated from GaAs/AlGaAs heterostructures with  two dimensional electron gas (2DEG) nominally \SI{90}{nm} below the surface. Quantum dots are formed by CrAu top gates depleting a shallow-etched mesa \cite{Gerster2018}. The charge detector is formed against the edge of a separate mesa and capacitively coupled to the central quantum dot via a floating gate \cite{Fricke2014}.
All measurements were performed in a dilution refrigerator at a base temperature of \SI{20}{\milli\kelvin} and \SI{0}{\tesla} external field. The charge detector signal is read out by rf reflectometry \cite{Schoelkopf1998}. Sinusoidal pulses generated by arbitrary waveform generators modulate the entrance barriers of the single electron pumps and drive the clock-controlled electron transfer \cite{Kaestner2015}. The drift-stability \replaced[id=R]{due to}{of the} control voltages is estimated to be better than $10^{-8}$. Charge transfer and detector readout are triggered in a sequence: (i) readout of the initial detector state, (ii) application of $t$ sinusoidal pulses to both pumps simultaneously, (iii) readout of the final detector state, (iv) reset by connecting the intermediate dot to source. The difference between  initial and final detector state yields the charge $x$ deposited on the central quantum dot by the burst transfer, providing raw data for subsequent statistical analysis.
Fisher's $p$-value for each experimentally measured $x$-resolved set of $N$ counts is defined as the probability of an equally or more extreme outcome under the null-hypothesis being tested (either the  baseline random walk or one of the two excess noise models with Dirichlet-distributed $P_{\pm}$); it is evaluated by Monte Carlo sampling as described in the Supplementary Notes \ref{sec:est_conf_region} and \ref{sec:dirichlet}.

\section*{Acknowledgements}
We acknowledge T. Gerster, L. Freise, H. Marx, K. Pierz, and T. Weimann for support in device fabrication, J. Valeinis for discussions. D.R. additionally acknowledges funding by the Deutsche Forschungsgemeinschaft (DFG) under Germany’s Excellence Strategy – EXC-2123 – 390837967, as well as the support of the Braunschweig International Graduate School of Metrology B-IGSM. M.K., A.A., and V.K are supported by Latvian Council of Science (grant no.~lzp-2018/1-0173). A.A. also acknowledges support by `Quantum algorithms: from complexity theory to experiment' funded under ERDF programme 1.1.1.5.

\begingroup
\renewcommand{\addcontentsline}[3]{}

\endgroup

\newcommand{\beginsupplement}{			% Label figures/tables with leading 'S'
	\setcounter{table}{0}
	\renewcommand{\thetable}{S\arabic{table}}
	\setcounter{figure}{0}
	\renewcommand{\thefigure}{S\arabic{figure}}
	\setcounter{secnumdepth}{2}
}

\renewcommand\bibsection{\section*{\refname}}

\onecolumngrid
\newpage
\section*{Supplemental information}\label{sec:supp}	

\beginsupplement

	\makeatletter
	\let\@seccntformatorig\@seccntformat
	\def\@seccntformat#1{%
		\ifnum0=\pdfstrcmp{#1}{section}%
		SUPPLEMENTARY NOTE \csname the#1\endcsname.{} %Well, could say Item \thesection:{} here as well... 
		\else
		\@seccntformatorig{#1}%
		\fi
	}

	\makeatother
	
	%\section{Supplemental}\label{sec:supp}

	\tableofcontents
	
	\addtocontents{toc}{\protect\setcounter{tocdepth}{2}}
	
	\newpage
	
	\section{Baseline model}\label{sec:basic_model}
	Consider a time-homogeneous discrete-time random walk on the set of integers which starts at 0 and at each step moves $ +1 $ with probability $ \probP  $, moves $ -1 $  with probability $ \probQ $ and stays at the same vertex with probability $ \probO $; here we assume $ \probP, \probQ, \probO \in (0,1) $,  $  \probQ+ \probO+ \probP = 1 $.
	
	To describe this process formally,  consider  a random variable $ \mb K = (K_{-1}, K_0, K_{+1}) $ following a multinomial distribution with $ t>0$ trials and three categories,  with associated probabilities $ \probQ $, $ \probO $ and $ \probP $, respectively. Then the random variable $X= K_{+1} - K_{-1}  $ corresponds to the position of the random walker after $ t $ steps, since all steps can be modeled with independent discrete random variables with three possible outcomes ($ -1 $, $ 0 $ and $ +1 $, respectively) and respective probabilities  $ \probQ $, $ \probO $ and $ \probP $. 
	First we show that the probability mass function of the discrete variable    $X  \in \{-t, -t+1, \ldots, t-1, t\} $   is given by  \eqref{eq:px} in the main text; i.e., let $ {p_{x}^t}:=   \Pr(X=x)  $, then
	\begin{claim}
		The probability mass function (PMF) of the variable $ X $   is
		\begin{equation}\label{eq:ptx_model0}
		{p_{x}^t} =
		\begin{cases}
		(\probP)^x (\probO)^{t-x} \, \binom{t}{x} \, {}_2F_1\left(\frac{x-t}{2},\frac{x-t+1}{2};x+1;\frac{4\probP\probQ}{\probO^2} \right),   &  x \geq 0,\\
		(\probQ)^{-x} (\probO)^{t+x} \, \binom{t}{-x} \, {}_2F_1\left(\frac{-x-t}{2},\frac{-x-t+1}{2};-x+1;\frac{4\probP\probQ}{\probO^2} \right), & x < 0,
		\end{cases}
		\end{equation}
		for $ x \in  \{-t, -t+1, \ldots, t-1, t\} $. 
	\end{claim}
	\begin{proof}
		Suppose that $ x\geq 0 $; then   the event $ X=x $, i.e., the event of the random walker being at the position $ x $ after $ t $ steps, is equivalent   to the event that the multinomially distributed  variable $ \mb K = (K_{-1}, K_0, K_{+1}) $ satisfies $ K_{+1} - K_{-1} = x $ (i.e., to the event that the random walker has moved $ K_{+1} $ steps to the right and $ K_{-1} = K_{+1} -x $ steps to the left). 
		Therefore $  \Pr(X=x) $ can be obtained from the multinomial distribution's PMF as 
		\[ 
		\Pr(X=x) = \sum_{\substack{\mb k: \\   k_1 -  k_{-1} = x \\  k_{-1} + k_0 + k_1 = t }} \Pr(\mb K = \mb k ).
		\]
		
		Since $\mb K $ follows a multinomial distribution with $ t>0$ trials and three categories  with respective  probabilities $ \probQ $, $ \probO $ and $ \probP $,  its PMF is given by
		\[ 
		\Pr(\mb K = \mb k ) = \frac{t!}{k_{-1}!  k_0!  k_1!}   \, \probQ^{k_{-1}}  \probO^{k_0}  \probP^{k_1}, 
		\]
		where $ \mb k := (k_{-1}, k_0, k_1) $ is a vector of nonnegative integers satisfying $ k_{-1}+k_0+k_1=t $.
		
		Notice that such $ \mb k $ can additionally satisfy $ k_1 - k_{-1} = x \geq 0 $ only if  $ \mb k $ is of the form
		$ \mb k  = (s,  t-x-2s,  x+s) $, for some nonnegative integer $ s $. Moreover, since $ k_0 \geq 0 $, we obtain the constraint $  t-x-2s \geq 0$, i.e., $ s \leq   0.5 (t-x) $. We conclude that all suitable vectors $ \mb k $ are parametrized by  a nonnegative integer $ s  $, which is upper-bounded  by $ 0.5 (t-x) $ (more precisely, the maximal valid $ s  $ value is  the floor function of $ 0.5 (t-x) $). For each such $ s $ the corresponding vector is  $ \mb k= (s,  t-x-2s,  x+s)   $ and the probability of the event $ \mb K = \mb k $ is  
		\[ 
		\Pr(\mb K =  (s,  t-x-2s,  x+s)   ) =  \frac{t!}{s!  (x+s)!  (t-x-2s)!}   \, \probQ^{s}  \probO^{t-x-2s}  \probP^{x+s} ,
		\]
		respectively.
		Thus  $ p_{x}^t:= \Pr(X=x) $ is simply the sum of these multinomial probabilities:
		\begin{align*}
		p_{x}^t & = \sum_{\substack{\mb k: \\   k_1 -  k_{-1} = x \\  k_{-1} + k_0 + k_1 = t }} \Pr(\mb K = \mb k )
		= \sum_{s = 0}^{ (t-x)/2}  \Pr \lr{ \mb K = (s, t-x-2s,  x+s)  }   \\
		& =    \sum_{s =0}^{ (t-x)/2}  \frac{t!}{s!  (x+s)!  (t-x-2s)!}   \, \probQ^{s}  \probO^{t-x-2s}  \probP^{x+s} .
		\end{align*}
		The latter quantity can be equivalently expressed as
		\begin{align*}
		p_{x}^t
		& =\probP^x \probO^{t-x}  \, 
		\binom{t}{x} \sum_{s = 0}^{ (t-x)/2}   \frac{ (t-x)!}{  (x+s)!  (t-x-2s)!  s!}   \,  \lr{\frac{\probP\probQ}{\probO^2}  }^s .
		\end{align*}
		Let us show that 
		\begin{equation}\label{eq:ptx_model0-1}
		\sum_{s = 0}^{ (t-x)/2}   \frac{ (t-x)!}{  (x+s)!  (t-x-2s)!  s!}   \,  \lr{\frac{\probP\probQ}{\probO^2}  }^s 
		=
		\, {}_2F_1\left(\frac{x-t}{2},\frac{x-t+1}{2};x+1;\frac{4\probP\probQ}{\probO^2} \right),
		\end{equation}
		which  will establish \eqref{eq:ptx_model0} and conclude the proof.
		
		We start by rewriting the LHS of \eqref{eq:ptx_model0-1}. Observe that
		$ (x+s)! =(x+1)_s $,
		where $ (a)_s  := a(a+1) \ldots (a+s-1) $ stands for the Pochhammer's symbol. Furthermore,
		\begin{align*}
		\frac{(t-x)!}{(t-x-2s)!  }
		&= 
		(t-x)(t-x-1) \ldots (t-x+1-2s) \\
		&=
		4^s    \lr{\frac{t-x}{2} - s+1 }_s   \lr{\frac{t-x-1}{2} - s+1 }_s\\
		&= 4^s \lr{\frac{x-t}{2}}_s\lr{\frac{x-t+1}{2}}_s 
		,
		\end{align*}
		where the last step applies the identity $ (-a)_s = (-1)^s (a-s+1)_s $. Therefore  
		\[ 
		\sum_{s = 0}^{ (t-x)/2}   \frac{ (t-x)!}{  (x+s)!  (t-x-2s)!  s!}   \,  \lr{\frac{\probP\probQ}{\probO^2}  }^s 
		=
		\sum_{s = 0}^{ (t-x)/2}   \frac{ \lr{\frac{x-t}{2}}_s\lr{\frac{x-t+1}{2}}_s }{  (x+1)_s   \,  s!}   \,  \lr{\frac{4\probP\probQ}{\probO^2}  }^s .
		\]
		The upper limit $  (t-x)/2 $ in the latter sum can be replaced with infinity, since the numerator $ \lr{\frac{x-t}{2}}_s\lr{\frac{x-t+1}{2}}_s $ is zero for the additional terms with $ s > 0.5(t-x) $.
		It remains to recognize now that the sum coincides with the definition of the Gaussian hypergeometric function ${}_2F_1$. We have verified \eqref{eq:ptx_model0-1},
		which concludes the proof of \eqref{eq:ptx_model0} when $ x\geq 0 $. The case $ x<0 $ follows from similar considerations.
	\end{proof}
	We note that   discrete distributions similar to $ X $   have been considered before. In particular, \cite{Zhang2011}  considers an analogue of our random variable $ X $ and computes $ p_0^t $ (termed ``return probability $ p_0(t) $'' in the paper). $ X $ is also closely related to the \emph{inverse trinomial distribution} \cite{shimizu1991,Aoyama05}, defined via a random walk on the line. Nevertheless, we are not aware of prior work establishing the PMF \eqref{eq:ptx_model0} of $ X $.

	The return probability for the random walk, $p_0^t$, is the probability of error-free electron transfer in the context of our benchmarking application, and hence can also be interpreted as transfer fidelity.  As long as the contribution of the return paths is negligible, $p_0^t$ decays exponentially, but for larger $t$ the exponential decay is modified. Below we derive an explicit asymptotics that characterizes both sides of this crossover.
	\begin{claim}
		\[
		p_0^{t} \approx
		\begin{cases}
		P_0^{t} \, , & t \ll (P_{+} {P_-})^{-1/2} \\
		\frac{\left ( P_0 + 2 \sqrt{P_{+} P_{-}}\right)^{1/2+t}}{
			(4 \pi \, t)^{1/2} (P_{+} \, P_{-})^{1/4} } \, , & t \gg (P_{+} {P_-})^{-1/2} 
		\end{cases} .
		\]
	\end{claim}
	\begin{proof}
		\newcommand{\zbar}{\zeta}
		By \eqref{eq:ptx_model0} we have
		\[ 
		p_0^t =  (\probO)^{t}  {}_2F_1    \left(\frac{-t}{2},\frac{-t+1}{2};  1; z \right), 
		\]
		where we denote $ z =4\probP\probQ / \probO^2   $.
		We start by observing that by a  quadratic transformation \cite[\href{https://dlmf.nist.gov/15.8.E13}{Eq.~15.8.13}]{DLMF}
		we have
		\begin{equation}\label{eq:claim2_eq1}
		{}_2F_1 \lr{  \frac{-t}{2},  \frac{-t+1}{2}  ;  1;  z} = {}_2F_1\lr{-t, 0.5; 1;\zbar} \cdot \lr{1 - 0.5 \zbar}^{-t},
		\end{equation}
		where the variable $ \zbar $ is defined by $ \frac{\zbar }{2- \zbar} = \sqrt z $, i.e.,
		\[ 
		\zbar =  \frac{2\sqrt z}{1 + \sqrt z}=  \frac{4\sqrt{P_{+} P_{-}}}{P_0 + 2 \sqrt{P_{+} P_{-}}}  
		\quad\text{and}\quad
		(1 - 0.5 \zbar)^{-1}  =  \frac{P_0 + 2 \sqrt{P_{+} P_{-}}}  {P_0}.
		\] 
		Using the equality \eqref{eq:claim2_eq1} we arrive at
		\begin{equation}\label{eq:claim2_eq2}
		p_0^t =  (P_0 + 2 \sqrt{P_{+} P_{-}})^{t} \,   {}_2F_1    \left( -t, 0.5;  1; \zbar \right).
		\end{equation}
		Since the hypergeometric function on the right hand side of \eqref{eq:claim2_eq2} is a degree-$t $ polynomial in the variable $ \zbar \ll 1 $, for small values of $ t $ we can approximate
		\[ 
		P_0 + 2 \sqrt{P_{+} P_{-}} \approx P_0
		\quad\text{and}\quad 
		{}_2F_1    \left( -t, 0.5;  1; \zbar \right) \approx 1,
		\]
		leading to the first part of the claim.

		Now we consider \eqref{eq:claim2_eq2} with fixed $ \zbar $ when $ t \to +\infty $. We apply an asymptotic expansion of the hypergeometric function in case of a large argument due to Erd\'elyi \cite[p.~77, Eq.~15]{Erdelyi53}, which exploits the relation between the (Gaussian) hypergeometric function $ {}_2 F_1 $ and the confluent hypergeometric function $ {}_1 F_1 $:
		\[ 
		{}_2 F_1 \lr{ -t, 0.5; 1; \zbar} \sim _{1}F_1(0.5; 1;  -t \zbar)  
		\sim \frac{\Gamma (1)  (  t\zbar )^{-0.5} }{\Gamma(0.5)}  \lr{ 1  + O\lr{ \lrv {t \zbar }^{-1}  }} 
		\sim  \frac{1}{   \sqrt {\pi t \zbar }} .
		\]
		Combining this with \eqref{eq:claim2_eq2} and substituting $ 	\zbar =   \frac{4\sqrt{P_{+} P_{-}}}{P_0 + 2 \sqrt{P_{+} P_{-}}}   $  gives us the second part of the claim.
	\end{proof}

	Finally, consider $ N $ independent observations of the random variable $ X$, i.e., i.i.d. random variables $ X_1,  \ldots , X_N \sim X $. Let $ Z_x = \lrv{\lrb{ j  \in \{1,\ldots, N\}  \; : \;   X_j = x }} $ be the number of times the value $ x \in \{-t,\ldots, t\} $ appears among these $ N $ observations. Then the random variable 
	\[ 
	\mb Z_{N,t} = (Z_{-t}, Z_{-t+1}, \ldots, Z_0, Z_1, \ldots, Z_t) 
	\] 
	follows  a multinomial distribution with $ N $ trials and $ 2 \, t+1 $ categories, labeled from $ -t $ to $ t $, and respective probabilities  $  {p_{x}^t}$.  When there is no ambiguity, this notation is simplified to $ \mb Z $.
	The probability to observe a particular vector $ \mb z \in \mbb N_0^{2t+1} $, $ \sum_{x=-t}^t z_x = N $ (where $ \mbb N_0 $ stands for the set of nonnegative integers) is
	\begin{equation}\label{model0_Z_pmf}
	\Pr(\mb Z = \mb z) =  N! \prod_{x=-n}^{n} \frac{ \lr{p_{x}^t}^{z_x} }{z_x!}.
	\end{equation}
	
	\section{Assessing \texorpdfstring{$ \probPQ $}{P+-} values from the experimental data}
	
	\subsection{Estimation of step-wise probabilities \texorpdfstring{$ \probPQ^{t} $}{P+-} by deconvolution}\label{sec:deconvolution-of--probpq-values}
	Under the assumption that the $\probPQ^t$-values  are independent of the position $x$ of the random walker, they can be extracted by deconvolution of $p_x^t$ and $p_x^{t+1}$.
	For that let us expand the model used so far and consider a random walk on the set of integers which at time $t$ performs transition $ x \mapsto x+j $ with probability $ P_j^t $, $ x,j \in \mbb Z $. Here $ P_j^t \in (0,1)  $ for all $ j $ and $ \sum_{j\in \mbb Z}P_j^t  =1  $.
	
	The experiment yields two vectors from $ \mbb R^{2m+1} $, $ m\in \mbb N $, {representing} the distributions  $ \mb p^t = \lr{ p_{-m}^t, \ldots,    p_{-1}^t, p_0^t, p_1^t, \ldots, p_m^t} $ and $ \mb p^{t+1} = \lr{ p_{-m}^{t+1}, \ldots,    p_{-1}^{t+1}, p_0^{t+1}, p_1^{t+1}, \ldots, p_m^{t+1}} $. We shall assume $ P_j^t = 0 $ for all $ j\in \mbb Z $ s.t. $ \lrv j \geq m $.
	The distribution  $  \mb p^{t+1}  $ represents the position of the random walker after $ t+1 $ steps and satisfies 
	\[ 
	p_x^{t+1}=\sum_{j \in \mbb Z} P_{j}^t p_{x-j}^t,
	\] 
	i.e., $  \mb p^{t+1}  =\mathbf{P}^t\ast\mathbf{p}^{t}$  is the discrete convolution of $ \mb P^t =\lr{P_{-m}^t, \ldots,P_{-1}^t,P_{0}^t,P_{1}^t,\ldots,P_{m}^t}$ and $ \mb p^t $. Therefore $ \mb P^t $ can be extracted by discrete deconvolution, which is performed as follows.
	
	Let $\mathfrak{p}^{t+1}$, $\mathfrak{p}^t$ and $\mathfrak{P}^t $ stand for the Fourier transform of $\mb p^{t+1}$, $\mb p^t$ and $\mb P^t $, respectively, then
	$$ \mathfrak{p}^{t+1}={\mathfrak{P}^t} \cdot\mathfrak{p}^t ,
	\quad \text{i.e.,} \quad  \mathfrak{P}^t_x = \frac{\mathfrak{p}^{t+1}_x}{\mathfrak{p}^t_x}
	\quad \text{for all } x.
	$$ 
	The vector $\mathfrak{p}^t$ is calculated from $p^t$ as
	\[
	\mathfrak{p}_x^t=\sum_{n=-m}^{m}p_n^t\exp{\left( -\imagI x\beta_n\right) }, \quad \text{where }\imagI=\sqrt{-1}\text{ and }\beta_n=\frac{2\pi n}{2m+1},
	\]
	similarly for $\mathfrak{p}^{t+1}$. 
	Now we can the get $\mathbf{P}^t$ by applying the inverse discrete Fourier transform to $\mathfrak{P}^t$:
	\[
	P_j^t
	=
	\frac{1}{2m+1}\sum_{n=-m}^{m}\mathfrak{P}^t_n\exp{\left( \imagI n\beta_j\right) }.
	\]
	Further details on how the deconvolution is performed and the uncertainty propagates can be found in \cite{Eichstaedt2016}.
	Now that we have extracted $\mathbf{P}^t$, we find for the experiment described in the main text, that $P_{|j|>1}^t\approx 0$ which allows us to approximate $\probP^t=P^t_{+1}$, and $\probQ^t=P^t_{-1}$.

	\subsection{Comparison between measured and predicted  \texorpdfstring{$ P_{\pm} $}{P+-}}
	\label{sec:compare_pumps}
	
	The characterization of the single electron pumps gives us their transport statistic $q_m^{(i)}$, which is the probability of pump $i \in\left\lbrace1,2\right\rbrace$ transporting $m \in \mathbb{Z}$ electrons. Assuming independence of simultaneous pump operation we can calculate the probability $P_x$ that charge on the island increases by $x \in \mathbb{Z}$ electrons (here $P_{\pm1}$ is equivalent to $\probPQ$ in \eqref{eq:px} in the  main text) as
	\begin{equation}\label{eq:Px_independent_pumps}
	P_x(\text{predicted})=\sum\limits_{m} q_{m+x}^{(1)}\cdot q_{m}^{(2)} \, .
	\end{equation}
	
	Table \ref{tab:node_dist} provides an example for agreement between measured and predicted values of \probPQ\ for non-interacting pumps.

	\begin{table} %[!h]
		%measurement SEP16-14\_X2Y2 03M046-03-M056
		\centering
		
		\begin{tabular}{@{}cS[table-column-width=12em]S[table-column-width=12em]}
			\toprule
			& \multicolumn{2}{c}{single pumps}\\
			\cline{2-3} \rule{0pt}{12pt}
			{$m$} & {$q_m^{(1)}$} & {$q_m^{(2)}$} \\
			\colrule
			0 & 0.00021 \pm 0.00005 &  3.6 \pm 2.8 e-05\\
			%\hline
			1 & 0.99982 \pm 0.00006 & 0.999975 \pm 0.000019\\
			%\hline
			2 &  0.0 \pm 2.0 e-05 &  0.0 \pm 1.2 e-05\\
			\botrule
		\end{tabular}
		\begin{tabular}{|cS[table-column-width=12em]S[table-column-width=12em]@{}}
			\toprule
			& \multicolumn{2}{c}{whole device}\\
			\cline{2-3} \rule{0pt}{12pt}
			{$x$} & {$P_x$ (measured)} & {$P_x$ (predicted)} \\
			\colrule
			\num{-1} &0.00012 \pm 0.00004 & 0.00021 \pm 0.00005\\
			%\hline
			\num{0} &0.99978 \pm 0.00005 & 0.99980 \pm 0.00006\\
			%\hline
			\num{1} & 0.0 \pm 2.1 e-05 &  3.6 \pm 3.4 e-05\\
			\botrule
		\end{tabular}
		
		\caption{Comparison between measured and predicted values of $P_x$, showing good agreement.}
		\label{tab:node_dist}
	\end{table}

	For the measurement in Table \ref{tab:stop_ramp} the waveform of the pump drive was changed from a low-frequency sinusoidal to a sharp voltage transient. Here we see a strong disagreement between the prediction of single pump characterization and the measurement of the \probPQ -values. This disagreement is caused by a strong shift of the operation point which occurs as soon as the pumps are operated simultaneously, indicating a strong correlation between the pumps.
	
	\begin{table}
		%measurement SEP16-14\_X2Y2 03M090 (stop ramp)
		\centering
		\begin{tabular}{@{}cS[table-column-width=12em]S[table-column-width=12em]}
			\toprule
			& \multicolumn{2}{c}{single pumps}\\
			\cline{2-3} \rule{0pt}{12pt}
			{$m$} & {$q_m^{(1)}$} & {$q_m^{(2)}$} \\
			\colrule
			
			0 & 0.00012 \pm 0.00016 & 0.00000 \pm 0.00014\\
			%\hline
			1 & 0.99974 \pm 0.00019 & 1.00000 \pm 0.00014\\
			%\hline
			2 & 0.00012 \pm 0.00017 & 0.00000 \pm 0.00014\\
			\botrule
		\end{tabular}
		\begin{tabular}{|cS[table-column-width=12em]S[table-column-width=12em]@{}}
			\toprule
			& \multicolumn{2}{c}{whole device}\\
			\cline{2-3} \rule{0pt}{12pt}
			{$x$} & {$P_x$ (measured)} & {$P_x$ (predicted)} \\
			\colrule
			-1 & 0.239 \pm 0.006 & 0.00012 \pm 0.00021\\
			%\hline
			0 & 0.776 \pm 0.005 & 0.99974 \pm 0.00024\\
			%\hline
			1 & 0.00012 \pm 0.00016 & 0.00012 \pm 0.00022\\
			\botrule
		\end{tabular}
		
		\caption{Disagreement between measured and predicted $P_x$-values for sharp-transient waveform.}
		\label{tab:stop_ramp}
	\end{table}

	\section{Model consistency testing}\label{sec:est_conf_region}
	{The experimental data} consist of   observations (actually, rebinned observations as described in Supplementary Note \ref{sec:conf-region})  of random variables $ \mb Z_{N_1, t_1} $,  $ \mb Z_{N_2, t_2} $,  \ldots ,   $ \mb Z_{N_L, t_L} $, for several different pairs $ (N_1, t_1) $, \ldots, $ (N_L,t_L) $, which, according to the model outlined in Supplementary Note \ref{sec:basic_model}, all share the same step probabilities $ (\probQ, \probP )$.  
	
	We consider the problem of determining if there is a parameter $ \probPQ  $ such that the experimental data do not contradict  the model, at the fixed significance level. More generally, we are interested in extracting  a region in the parameter space such that  the experimental data do not contradict  the model for each choice of the parameter from the region; for brevity, we will refer to this region as \emph{consistency region}. It should be stressed that this approach is different from \emph{parameter estimation problem}, in that here we are interested in parameter values which cannot be statistically rejected as incompatible with the data, whereas the parameter estimation techniques deal with estimating the values of the parameters in some fashion, e.g., by finding the values of parameters under which the experimental data are most probable under the assumed model.
	
	The problem of testing  consistency of the model with a specific parameter value is twofold: since the data correspond to several pairs of $ (N,t) $, with different parameters $ N,t $ but the same  step probabilities $ (\probQ, \probP )$, there are two questions to be asked:
	\begin{enumerate}
		\item Are the data for the particular  value of $ (N,t) $ consistent with the model for some parameter $ \probPQ  $?
		\item Are all the data  consistent with the model for some fixed value of $ \probPQ  $?
	\end{enumerate}
	We start by testing consistency with the model  in case of an observation of  $ \mb Z_{N,t} $ for a single pair $ (N,t) $.
	
	\subsection{Fisher's  significance testing}\label{sec:fisher}
	%From the experimental data we have obtained an observation $ \mb z $ of  the random variable $ \mb Z $, with prescribed parameters $ t,N $ but unknown probabilities  $ \probQ, \probO, \probP $.
	
	Let  $ \mb z_0 $ be an {observation} of  the  random variable $ \mb Z := \mb Z_{N,t} $, with prescribed parameters $ t,N $ but unknown probabilities  $ \probQ, \probO, \probP $.
	
	We employ Fisher's  significance testing framework in order to extract the consistency regions for the parameter  $ {\boldsymbol\theta} = (\probQ, \probO, \probP )$. In its simplest form, a Fisherian test formulates \cite{Christensen2005} a single hypothesis, the null hypothesis $ H_0 $, which specifies the null distribution (i.e., in our case $ H_0 :  {\boldsymbol\theta} = {\boldsymbol\theta}^* $ for some fixed $ {\boldsymbol\theta}^* $); then a certain test statistic $ T $ is computed from the observation $ \mb z _0$, leading to a value $ T(\mb z_0) $. The $ p $-value of the test is the tail probability of $ T(\mb Z) $ under $ H_0 $.  In our setting,  the test statistic will be non-negative and smaller values will indicate stronger disagreement with the null hypothesis. Then the $ p $-value of the test is
	\[ 
	p(\mb z_0)= \sum_{\substack{\mb z :  \\  T(\mb z) \leq T(\mb z_0)}} \Pr (\mb Z=\mb z),
	\]
	where $ \Pr(\mb Z = \mb z) $ stands for the probability of  the event $ \mb Z = \mb z $ under the null hypothesis and the sum is over all those values $\mb z$ of the random vector $ \mb Z $ that satisfy $T(\mb z) \leq T(\mb z_0)$
	In the Fisher's  significance testing framework  the $ p $-value is interpreted as ``a measure of extent to which the  data do not contradict the model'' \cite[p.122]{Christensen2005}. Therefore Fisher's  significance testing allows to check if $H_0$ must be rejected (at the chosen significance level) for the particular value $\boldsymbol \theta^*$; next, we shall employ Fisher's  significance testing to extract the region of those $\boldsymbol \theta$ values {for which the respective} $H_0$ cannot be rejected, see Supplementary Note \ref{sec:conf-region}.

	The problem of testing  whether  the parameters of a multinomial distribution equal specified values has been well-investigated \cite{Conover1972,Smith1981,Cressie1984,Jann2008}. The common approaches (such as Pearson's $ \chi^2 $ test, $ G^2 $ test or power-divergence test \cite{Cressie1984} which subsumes the former tests) are asymptotic tests which  
	can be highly biased. This is due to the fact that under the null  hypothesis the random variable $ X $ has vanishingly small tail probabilities  (and the  actual observed samples $ \mb z $ have zero observed counts in the respective positions).  This phenomenon makes the asymptotic tests ill-suited for the actual data.

	An alternative to the aforementioned tests is the exact multinomial test \cite{Cressie1984}, which enumerates all possible multinomial outcomes; its test statistic $ T $ is the probability of obtaining the particular outcome   under the null hypothesis. Then
	the $ p $-value of the test is
	\[ 
	\sum_{\substack{\mb z :  \\  \Pr(\mb Z=\mb z) \leq \Pr(\mb Z=\mb z_0)}} \Pr (\mb Z=\mb z).
	\]
	
	However, the exhaustive enumeration  quickly becomes computationally intractable as $ N $ grows.
	We instead apply a Monte Carlo test (proposed in \cite{Barnard1963}, see also \cite{Hope1968, Besag1992,Jann2008}),  which can be seen as an extension of the exact multinomial test.  In the  Monte Carlo hypothesis testing  procedure, a large number (say, $ \Nsim $)  samples from the multinomial distribution under the null hypothesis are simulated; for each sample $ \mb z $ the test statistic $ \Pr (\mb z)  $ is calculated (i.e., the probability to draw $ \mb z $ from the distribution $ \mb Z $ under the null hypothesis). Let $ k $  be the number of samples for which the test statistic is at least as extreme as for the  observed vector $ \mb z_0 $ (i.e., the number of samples $ \mb z $ for which $ \Pr(\mb z) \leq \Pr (\mb z_0) $). Then the $ p $-value of the test is $ (k+1)/(\Nsim + 1) $.

	\subsection{Consistency regions}\label{sec:conf-region}
	Since $ \probO = 1-\probP-\probQ $,  the Monte Carlo tests are applied to extract $ 95\% $  consistency region for   the pair $ (\probQ, \probP )$. This region is defined as the set  of all admissible $ (\probQ, \probP )$  values for which the $ p $-value obtained by  testing the hypothesis $ H_0 :  {\boldsymbol\theta} =(\probQ,  (1-\probQ- \probP),\probP) $ is at least $ 0.05 $.
	
	In practice, since the observed vector $ \mb z_0 $ has many zero entries (as $  N$ is too small to observe ``$ X=x $''  when $ \lrv x $ is large) and, since the experimental data is limited to small $ \lrv x $, the data are rebinned. Depending on the dynamical range of the detector and available computational resources, we consider a random variable $ \tilde {\mb Z}  = (\tilde Z_{-3}, \tilde Z_{-2}, \tilde Z_{-1}, \tilde Z_0, \tilde Z_1, \tilde Z_2, \tilde Z_3)  $ instead of the random variable $ \mb Z $  where 
	\begin{equation}\label{eq:rebinning}
	\tilde Z_{-3} = \sum_{x \leq -3} Z_x, \quad
	\tilde Z_{3} = \sum_{x \geq 3} Z_x,
	\text{ and }
	\tilde Z_x= Z_x ,\    \lrv{x} \leq 2,
	\end{equation}
	and perform the aforementioned  tests against an observation $ \tilde {\mb z}_0 $ of $ \tilde {\mb Z} $.
	Further on, this subtlety will be assumed implicitly, i.e., when talking of the random variable $\mb Z$ or its observation $\mb z_0$, the rebinned counterparts $\tilde{\mb Z}$ and  $\tilde{\mb z}_0$ are to be understood.
	
	\subsection{Combining  the  \texorpdfstring{$ p$}{p}-values}\label{sec:comb_p-val}
	The discussion above attempts to answer if the data are consistent with some $\probPQ$, for a particular value of $ (N,t) $; the challenge now is to combine the statistical tests done for all $ L $ pairs of $ (N,t) $. While for each fixed pair $ \lr{N,t} $ the $ 95\% $  consistency region  can be constructed from the observation of the respective $ \mb Z_{N,t} $, the goal   is to obtain a global measure of discrepancy between the data and the hypothesis $ H_0 :  {\boldsymbol\theta} =(\probQ,  (1-\probQ- \probP),\probP) $, taking into account the observations for all  pairs $ (N,t) $.
	
	This task can be viewed as the problem of combining several independent $ p $-values, which arises in meta-analysis \cite{Borenstein2009}. When testing a true point null hypothesis  and the test statistic is absolutely continuous, it can be shown that the $ p $-values under the null hypothesis are uniformly distributed in $ [0,1] $. This allows to apply, e.g., Fisher's method of testing {uniformity} \cite{Fisher1932} (for an overview of other ways to combine $ p $-values, see \cite[Appendix A]{Winkler2016}). In our case both the random variables $ \mb Z_{N,t} $  and the test statistic are discrete, thus  under the null hypothesis all $ p $-values obtained for each pair  $ (N,t) $   only   {approximate the uniform distribution}.    Fisher's method is used to approximately determine the combined $ p $-value, even though in case of sparse discrete distribution this approximation  may \cite{Mielke2004} yield conservative results. 
	
	{In practice}, due to the computational cost involved with computing the combined $ p $-value, this global consistency test is only performed for a single value of $ \boldsymbol \theta $. The value $ (\probP,\probQ)  =(\numlistDef{2.130664e-5;6.924426e-5})$ we performed the combined test on is the one under which the observed data are most probable, i.e.,  the maximum likelihood estimate, see Supplementary Note \ref{sec:MLE}. However, the combined $ p $-value $ \numlistDef{2.230876680796699e-06} $ means that $H_0  $ needs to be rejected; also visually (see Figure \ref{fig:fig2}c, triangles) it is clear that the distribution of $ p $-values is far from uniform. Hence one concludes that this model with fixed $ \probPQ $ for all pairs $ (N,t) $ is incompatible with the experimental data.
	
	\subsection{Maximum likelihood   estimation}\label{sec:MLE}
	The preceding discussion tries to determine if the data    contradict    the model, within the given level of significance. However, if one only tries to find the most suitable choice of parameters $ \probPQ $, a natural approach is to maximize the likelihood function, i.e., (in case of a single observation for a single pair $ (N,t) $) maximize the expression in \eqref{model0_Z_pmf}, with $ z_x $ being the actual observed values, with respect to the unknown parameters  $ \probPQ $. The task is equivalent to maximizing the logarithm of the likelihood, 
	\[ 
	\ell_{N,t}(\boldsymbol\theta) = \ln \Gamma(N+1) +  \sum_{x=-n}^{n}  \lr{z_x \ln \lr{p_{x}^t(\boldsymbol \theta)}  - \ln \Gamma (z_x+1)},
	\]
	where $ p_{x}^t(\boldsymbol \theta) $ stands for the RHS in \eqref{eq:ptx_model0} and $ \ln \Gamma $ is the natural logarithm of the gamma function.

	Since the observations across the $ L $ different pairs $ (N_i,t_i) $ are assumed to be independent, the joint probability of observing the complete data is the product of individual probabilities for each separate $ (N_i,t_i) $, i.e., the global log-likelihood function to be maximized is
	\[ 
	\ell (\boldsymbol\theta)  =  \sum_{i=1}^L   \ell_{N_i,t_i}(\boldsymbol\theta).
	\] 
	Maximizing this function over the standard 2-simplex using the experimental data gives  the maximum likelihood estimate $ (\probP,\probQ)  =(\numlistDef{2.130664e-5;6.924426e-5})$.

	\section{Dirichlet distribution-based random-walk models}\label{sec:dirichlet}
	Further we consider the case when the step probabilities $ \probQ $, $ \probP $ are themselves random variables.
	We assume that $  (\probQ,  \probO,\probP)$   follows a Dirichlet distribution, which is   \cite{Ng2011} ``one of the key multivariate distributions for random vectors confined to the simplex''.    
	The Dirichlet distribution also becomes important when the observed data are superficially similar to the multinomial distribution but exhibit more variance than   the multinomial  distribution permits. As authors in \cite[p.199]{Ng2011} note, ``One possibility of this kind of extra variation is that the multinomial probabilities'' are not constant across the trials and the vector of probabilities can be interpreted as a random vector in the standard simplex; in this case the Dirichlet distribution is a convenient choice, resulting in a compound probability distribution, the Dirichlet-multinomial distribution \cite[Definition 6.1]{Ng2011}.

	The Dirichlet distribution on the standard 2-simplex  $ \Delta^2 $ with positive parameter vector $ \boldsymbol \alpha=(\alpha_0, \alpha_1,\alpha_2) $,  denoted by $ \Dir(\boldsymbol\alpha) $,  is a probability distribution with 
	\cite[Definition 2.1]{Ng2011} the density function  
	\begin{equation}\label{eq:Dir_PDF}
	f_{\boldsymbol\alpha}({\boldsymbol\theta} ) = \frac{\prod_{i=0}^{2}\Gamma(\alpha_i)}{\Gamma(\sum_{i=0}^{2}\alpha_i)}  \theta_0^{\alpha_0-1} \theta_1^{\alpha_1-1} \theta_2^{\alpha_2-1}, \quad {\boldsymbol\theta}=(\theta_0, \theta_1, \theta_2) \in \Delta^2.	
	\end{equation}
	The mean  value and the variance of $ \theta_i $, $ i=0,1,2 $,   are given by
	\[ 
	\Expect(\theta_i)  = \frac{\alpha_i}{\sum_{j=0}^2 \alpha_j }     = :\tilde \alpha_i ,
	\quad
	\Var(\theta_i)   =  \frac{\tilde \alpha_i (1- \tilde \alpha_i)}{1+ \sum_{j=0}^{2}\alpha_j  },
	\] 
	respectively,  i.e., the mean value of $ \theta_i $ is proportional to the parameter $ { \alpha_i} $, but   the variance of $ \theta_i $ decreases as   $ \sum_{i=0}^{2}\alpha_i $ is increased. This allows to employ  the Dirichlet distribution  to model the scattering of the vector  $ (\probQ, \probO, \probP ) \in \Delta^2$ around its mean value  with a single additional parameter  characterizing the magnitude of the scattering.

	We proceed by considering two extensions of the baseline model, one where the variable $ {\boldsymbol\theta}=(\probQ, \probO, \probP ) $ is chosen independently for each of the $ N $ separate random walks, and another where $ {\boldsymbol\theta} =(\probQ, \probO, \probP ) $  is the same for all $ N $ random walks (but another $ {\boldsymbol\theta}\sim \Dir({\boldsymbol\alpha}) $ is independently drawn if either $ N $ or $ t $ is changed).

	\subsection{Model 1 (fast fluctuator)}\label{sec:model-1} 
	\subsubsection{Model description} 
	Let $ {\boldsymbol\alpha} = (\alpha_{-1} , \alpha_0, \alpha_1) $ be a fixed vector of positive parameters.   For each pair $ (N,t) $ we consider the following process:
	\begin{itemize}
		\item repeat $ N $ times:
		\begin{itemize}
			\item choose a random vector  $ {\boldsymbol\theta} =(\probQ,  \probO,\probP) \sim  \Dir({\boldsymbol\alpha})  $ (independently each time);
			\item perform $ t $ steps of  the random walk  with the respective step probabilities $ (\probQ,  \probO,\probP)  $;
			\item observe the  position of the random walker $X  \in \{-t, -t+1, \ldots, t-1, t\} $;
		\end{itemize}
		\item  given the $ N $   observations   $ X_1,  \ldots , X_N  $, denote $ Z_x = \lrv{\lrb{ j  \in \{1,\ldots, N\}  \; : \;   X_j = x }} $ and define the random variable 
		\[ 
		\mb Z_{N,t} = (Z_{-t}, Z_{-t+1}, \ldots, Z_0, Z_1, \ldots, Z_t) .
		\]    
	\end{itemize} 
	This model corresponds to choosing step probabilities $ \probQ, \probP $ independently for each  repetition of a random walk of a fixed length $t$. This way the random variable $ \mb Z_{N,t} $  again has multinomial distribution, but now with modified (compared to the baseline model) probabilities incorporating the underlying Dirichlet distribution.
	
	\subsubsection{Probability mass function} 
	To describe this process more formally, for each pair $ (N,t) $ let $ \mb K = (K_{-1}, K_0, K_{+1}) $ have the Dirichlet-multinomial distribution with $ t>0 $ trials and parameter $ {\boldsymbol\alpha} = (\alpha_{-1} , \alpha_0, \alpha_1) $. Define a random variable $X= K_{+1} - K_{-1}  $, supported in the set $ \{-t,  -t+1, \ldots, t-1, t\} $; denote $ {p_{x}^t :=   \Pr(X=x)  } $ and define a multinomial   variable $ \mb Z_{N,t} $ with $ N $ trials, $ 2t+1 $ categories (from $ -t $ to $ t $) and the respective  probabilities $ p_x^t $, $ x  \in  \{-t,  -t+1,  \ldots, t-1, t\} $.

	Since $\mb K $ follows the Dirichlet-multinomial distribution, its PMF   (for a vector of nonnegative integers $ \mb k = (k_{-1}, k_0, k_1) $ s.t. $ k_{-1} + k_0 + k_1 = t $) satisfies %\cite[Eq.(6.4)]{Ng2011}
	\cite[Eq. 35.152]{Johnson97}
	\[ %\begin{equation}\label{eq:K_probs_model1}
	\Pr(\mb K = \mb k) 
	=
	\frac{t! \, \Gamma(\alpha_\bullet )}{\Gamma(t+\alpha_\bullet)}   \frac{\Gamma(k_{-1} + \alpha_{-1})\Gamma(k_{0} + \alpha_0)\Gamma(k_{1} + \alpha_1)}{k_{-1}!\, k_0!\,  k_1! \Gamma(\alpha_{-1})\Gamma(\alpha_0)\Gamma(\alpha_1)},
	\]
	where we denote $\alpha_\bullet := \sum_i \alpha_i$.
	Notice that if we keep the fractions $\theta_i := \frac{\alpha_i}{\alpha_\bullet}$ fixed, then in the limit $\alpha_\bullet\to \infty$ the random variable $\mb K$ becomes multinomially distributed, i.e.,
	\[
	\Pr(\mb K = \mb k) 
	\xrightarrow[\alpha_\bullet \to \infty]{}
	\frac
	{t!}
	{k_{-1}!\, k_0!\,  k_1!  }   
	\theta_{-1}^{k_{-1}}\theta_{0}^{k_{0}}\theta_{1}^{k_{1}}
	.
	\]
	This  follows easily from the gamma function property $ \Gamma (k+a ) \sim \Gamma (a)a^{k }$ as $a\to \infty$.

	Henceforth,
	\begin{align}
	p^t_x & =  \sum_{\substack{\mb k: \\   k_1 - k_{-1} = x\\ k_{-1} +k _0 + k_1 = t}} \Pr(\mb K = \mb k)  =\sum_{l = \max \{0,-x\}}^{ (t-x)/2}  \Pr \lr{ \mb K = (l, t-x-2l,  x+l)  } \notag  \\
	& =
	\frac{t! \, \Gamma(\alpha_\bullet)}{\Gamma(t+\alpha_\bullet)  \prod_i  \Gamma(\alpha_i)}  
	\, 
	\sum_{l = \max \{0,-x\}}^{ (t-x)/2}  \frac{ \Gamma(l + \alpha_{-1})  \Gamma(t-x-2l+ \alpha_{0})  \Gamma(x+l + \alpha_{1})  }{ l! (t-x-2l) ! (x+l)!  }. \label{eq:ptx_model1} %\Pr \lr{ \mb X = (l, n-j-2l,  j+l)  }   
	\end{align}
	Observe that keeping the fractions $ \frac{\alpha_i}{\alpha_\bullet}$ fixed and letting $\alpha_\bullet\to \infty$ makes the probabilities  $p^t_x$ given by \eqref{eq:ptx_model1} tend to the respective probabilities given by \eqref{eq:ptx_model0} (with $\probQ=\alpha_{-1}/\alpha_\bullet$ and $\probP=\alpha_1/\alpha_\bullet$).
	
	After $ N $ independent observations  the   multinomial vector $ \mb Z_{N,t} $ is obtained, 
	supported in the set  $ \lrb{ \mb z \in \mbb N_0^{2t+1}  : \sum_{x=-t}^t z_x = N}  $, with 
	\begin{equation}\label{eq:Z_PMF_model1}
	\Pr(\mb Z_{N,t}  = \mb z) =  N! \prod_{x=-n}^{n} \frac{ \lr{p_{x}^t}^{z_x} }{z_x!}, \quad  \mb z \in \mbb N_0^{2t+1}  ,\ \sum_{x=-t}^t z_x = N.
	\end{equation}
	
	The variable $ \mb Z_{N,t} $  still has the multinomial  distribution, as in the baseline model, and Eq.~\eqref{eq:Z_PMF_model1} is the same as \eqref{model0_Z_pmf} but with $p_{x}^{t}$ given by \eqref{eq:ptx_model1}.
	However, in contrast to the baseline model,  the vector $ \mb K $ has the Dirichlet-multinomial distribution instead of the multinomial distribution as before. 
	That, in turn, implies that the probabilities   $ p_x^t $ are not calculated from \eqref{eq:ptx_model0}, but   given by \eqref{eq:ptx_model1} instead. In effect, $ \mb Z_{N,t} $  is a multinomial distribution, but different probabilities associated with its categories, when compared to the baseline model.

	\subsubsection{Consistency testing} 
	Consistency of this model is tested similarly as in the baseline case:
	\begin{itemize}
		\item For each particular pair $ (N,t) $, we perform a Fisherian test of the hypothesis $ H_0 :  {\boldsymbol\alpha} = {\boldsymbol\alpha}^*  $ for some fixed $ {\boldsymbol\alpha}^* $, given  an observation $ \mb z_0 $ of $ \mb Z =\mb Z_{N,t}$. The test is again conducted in the Monte Carlo manner  as described previously,   by drawing $ \Nsim $ samples from the multinomial distribution under the null hypothesis and extracting the $ p $-value as $ (k+1)/(\Nsim + 1) $. Here $ k $ indicates the number of the simulated samples $ \mb z $ satisfying $ \Pr( \mb Z  = \mb z) \leq \Pr ( \mb Z  = \mb z_0) $. 
		\item Consistency of the model taking into account all $ L $ different pairs $ (N,t) $ is done by combining the $ L  $ obtained  $ p $-values, via Fisher's method of testing uniformity.
	\end{itemize}
	The value  $ {\boldsymbol\alpha}^* $ to be tested in the previous step is again   the maximum likelihood estimate, obtained by maximizing  the function
	\[ 
	\ell (\boldsymbol\alpha)  =  \sum_{i=1}^L   \ell_{N_i,t_i}(\boldsymbol\alpha),
	\]
	where
	\[ 
	\ell_{N,t}(\boldsymbol\alpha) = \ln \Gamma(N+1) +  \sum_{x=-n}^{n}  \lr{z_x \ln \lr{p_{x}^t(\boldsymbol \alpha)}  - \ln \Gamma (z_x+1)},
	\]
	and $ p_{x}^t(\boldsymbol \alpha) $ is given by the RHS of \eqref{eq:ptx_model1}. 
	Maximizing this function over the parameter space  using the experimental data gives  the maximum likelihood estimate $ \boldsymbol\alpha^*  =(\numlistDef{9.07719849e+01;1.30960421e+06;2.78484341e+01})$. However, the combined $ p $-value $ \numlistDef{2.0768436977413293e-06} $ 
	again
	indicates that $H_0  $ needs to be rejected; as it is seen in Figure \ref{fig:fig2}c (squares), the distribution of $ p $-values still remains far from uniform. Consequently, this model is also incompatible with the experimental data.
	
	\subsection{Model 2 (slow drift)} \label{sec:model-2}
	\subsubsection{Model description} 
	Let again $ {\boldsymbol\alpha} = (\alpha_{-1} , \alpha_0, \alpha_1) $ be a vector of positive parameters.  Now we consider the following process  for each  pair $ (N,t) $:
	\begin{itemize}
		\item  choose    a   random vector $ {\boldsymbol\theta} =(\probQ,  \probO,\probP) \sim  \Dir({\boldsymbol\alpha})  $; %   (independently across different pairs $ (N,t) $);
		\item repeat $ N $ times:
		\begin{itemize}
			\item perform $ t $ steps of  the random walk  with the respective step probabilities $ (\probQ,  \probO,\probP)  $;
			\item observe the  position of the random walker $X  \in \{-t, -t+1, \ldots, t-1, t\} $;
		\end{itemize}
		\item  given the $ N $   observations   $ X_1,  \ldots , X_N  $, denote $ Z_x = \lrv{\lrb{ j  \in \{1,\ldots, N\}  \; : \;   X_j = x }} $ and define the random variable 
		\[ 
		\mb Z_{N,t} = (Z_{-t}, Z_{-t+1}, \ldots, Z_0, Z_1, \ldots, Z_t) .
		\]   
	\end{itemize}
	This way, the vector $ {\boldsymbol\theta} =(\probQ,  \probO,\probP) \sim  \Dir({\boldsymbol\alpha})  $   is drawn independently across different pairs $ (N,t) $, yet  for each particular $ (N,t) $ it is fixed for all $ N $ random walks (the $ N $ random walks are assumed to be conditionally independent given $ \boldsymbol\theta $). The  resulting random variable $ \mb Z_{N,t} $ has a discrete compound distribution, akin to the Dirichlet-multinomial distribution; however, $ \mb Z_{N,t }$ is not multinomially distributed anymore.
	
	Technically, the key difference from the previous model is that all $ N $ random walks use the same (randomly drawn from $ \Dir({\boldsymbol\alpha}) $) vector $ {\boldsymbol\theta} $, therefore marginalization of $ {\boldsymbol\theta} $ happens only after forming the counts vector $ \mb Z_{N,t} $. In contrast, in the previous model the Dirichlet variable is marginalized after forming the vector $ \mb K $, resulting in  the Dirichlet-multinomial distribution for $ \mb K $ and a standard multinomial variable $ \mb Z_{N,t} $.
	
	\subsubsection{Probability mass function} 
	To characterize the model  more formally, for each pair $ (N,t) $ and a fixed vector $ {\boldsymbol\theta} =(\probQ,  \probO,\probP)    $ let $ p_x^t ({\boldsymbol\theta})$, $ \lrv x \leq t $, be defined as in the RHS of \eqref{eq:ptx_model0}.
	The random variable $ \mb Z_{N,t} $ is defined by compounding the multinomial distribution \eqref{model0_Z_pmf} with the Dirichlet distribution $ \Dir({\boldsymbol\alpha}) $, i.e.,  $ \mb Z_{N,t} $ is  supported in the set  $ \lrb{ \mb z \in \mbb N_0^{2t+1}  : \sum_{x=-t}^t z_x = N}  $ and its PMF is  obtained by marginalizing over the Dirichlet variable: for $ \mb z \in \mbb N_0^{2t+1} $  such that $ \sum_{x=-t}^t z_x = N $,
	\begin{equation}\label{eq:Z_PMF_model2}
	\Pr(\mb Z_{N,t}  = \mb z) 
	=  
	\frac{N!}{ \prod_{x=-n}^{n} z_x! }
	\int_{\Delta^2 }  \prod_{x=-n}^{n}  \lr{p_{x}^t({\boldsymbol\theta})}^{z_x}  f_{\boldsymbol\alpha}({\boldsymbol\theta} )  \, \mathrm{d}{\boldsymbol\theta},
	\end{equation}
	where the integration is over the standard 2-simplex  $ \Delta^2 $ and
	\begin{equation}\label{eq:Dir_PDF_SD}
	f_{\boldsymbol\alpha}({\boldsymbol\theta} ) 
	= 
	\frac{\prod_{i=-1}^1  \Gamma(\alpha_{i}) \theta_i^{\alpha_i  - 1}}{\Gamma(\alpha_\bullet)},
	\quad  
	\alpha_{\bullet} := \alpha_{-1}+\alpha_0 + \alpha_1,
	\end{equation}
	is the PDF of the Dirichlet distribution (adapted from \eqref{eq:Dir_PDF}). It is worth mentioning that since only the parameters $\probPQ,\probO$ are chosen from the Dirichlet distribution, instead of all $2t+1$ event probabilities associated to the multinomial distribution, the resulting compound distribution is not Dirichlet-multinomial.

	\subsubsection{Consistency testing} \label{sec:model-2-consistency}
	Given an observation $ \mb z_0 $ of $ \mb Z_{N,t} $, we again perform Monte Carlo test of the hypothesis $ H_0 :  {\boldsymbol\alpha} = {\boldsymbol\alpha}^*  $, for some fixed $ {\boldsymbol\alpha}^* $.
	However, now the probability $ \Pr(\mb Z_{N,t}=\mb z) $ has the complicated analytical form \eqref{eq:Z_PMF_model2}, which is difficult to compute numerically. Therefore also $ \Pr (\mb Z_{N,t} = \mb z) $ is estimated via Monte Carlo approximation, i.e.,  for the  particular parameter $ {\boldsymbol\alpha}^* $ and the  observed vector $ \mb z_0 $ we
	\begin{itemize}
		\item draw $ \Nsim $ independent samples $ {\boldsymbol\theta} \in \Delta^2 $ from $ \Dir({\boldsymbol\alpha}^*) $;
		\item for each of the sampled vectors $ {\boldsymbol\theta}=(\probQ,  \probO,\probP)  $ draw a sample $ \mb z $ from the multinomial distribution specified by \eqref{model0_Z_pmf} (where the probabilities $ p_x^t $ are computed using the sampled values $ \probQ, \probP $).
		\item  This way $ \Nsim  $ vectors $ \mb z_1 $, \ldots, $ \mb z_{\Nsim} $ are obtained, among them many may coincide. Suppose that    $ m $  distinct vectors $ \mb z_1' $, \ldots, $ \mb z_m' $ were obtained, with their respective frequencies $ k_1$, $ k_2 $, \ldots, $ k_m $,  $ \sum_i k_i = \Nsim $. We can assume that $ k_1 \leq k_2 \leq \ldots \leq k_m $.
		
		\item Suppose that  $ \mb z'_j $ coincides with the actual observation $ \mb z_0 $, and (provided that $ j<m $) $ k_j  < k_{j+1} $; then the $ p $-value of the test is declared $ (k+1)/(\Nsim + 1) $, where $ k: = k_1 + k_2 + \ldots + k_j $. In case $ \mb z_0 $ does not occur among the $ \Nsim  $  obtained vectors, the $ p $-value is declared 0. 
	\end{itemize}
	By employing the outlined procedure, we can perform  consistency  testing similarly as before. Consistency of the model taking into account all $ L $ different pairs $ (N,t) $ is done by combining the $ L  $ obtained  $ p $-values, via Fisher's method of testing uniformity.

	The value  $ {\boldsymbol\alpha}^* $ to be tested now is {found}   differently, compared to the previous models. This is due to the fact that  the probabilities $ \Pr( \mb Z  = \mb z) $ are estimated only approximately via Monte Carlo, which complicates maximizing the likelihood function.
	
	Instead, we fix the fractions $ \frac{\alpha_i}{\alpha_\bullet}$ to the best values of $ \probPQ $ found in the baseline model (Supplementary Note \ref{sec:MLE})  and optimize the parameter $\alpha_\bullet $, i.e., $ \boldsymbol \alpha $ is in form $ \alpha_\bullet\cdot\ (\probQ, \probO, \probP) $, where $ (\probP,\probQ)=(\numlistDef{2.130664e-5;6.924426e-5})$. The cost function associated with $ \alpha_\bullet $ is 
	\begin{align} \label{eq:distanceMeasure}
	C(\alpha_\bullet) = \sum_{i=1}^L \lrv{p_{(i)}  - \frac{i}{L}  },
	\end{align}
	where $ p_{(i)} $ stands for the $ i $-th smallest value among $ p_1, p_2, \ldots, p_L $, where the latter are the $ p $-values returned by the $ L $ tests of the hypothesis $ H_0 :  {\boldsymbol\alpha} =\alpha_\bullet\cdot\ (\probQ, \probO, \probP) $. In other words, the cost function measures the distance between the empirical distribution function of $ p $-values and the line corresponding to the cumulative distribution function corresponding to the uniform distribution.

	The {minimization}  of the cost function over $ \alpha_\bullet $ estimates  the optimal parameter to 
	$\boldsymbol\alpha^*=(\numlist[list-final-separator = {, },round-mode = places,round-precision = 2]{2.427195e3;3.504960e7;7.468530e2})$, with precision limited by numerical expense of Monte Carlo trials for $p_{(i)}$. The combined $ p $-value  equals $ \numlistDef{0.7116902505502581} $,  therefore  the null hypothesis  $ H_0 :  {\boldsymbol\alpha} = {\boldsymbol\alpha}^*  $ cannot be rejected. Also, as   it is seen in Figure \ref{fig:fig2}c (diamonds), the distribution of $ p $-values  visually conforms to the uniform. Henceforth, the experimental data do not contradict this model.
	
	\section{Excess noise simulations}\label{sec:noise-sim}

	The purpose of this Supplementary Note is to illustrate that our  excess noise models can
	be consistent with and give reasonable estimates of realistic  parametric variability one expects from an ensemble of TLFs (multi-timescale $1/f$-noise), or from a single but strong TLF (bimodal excess noise), despite a  generic  Dirichlet distribution and a single correlation timescale underpinning the ``fast-fluctuator'' and the ``slow-drift'' tests.
	
	First we define a model that simulates fluctuating environment of the real experiments, and then analyze the simulated error counts  using  the same statistical tests as applied in the main text and  described in Supplementary Notes \ref{sec:est_conf_region}, \ref{sec:model-1} and \ref{sec:model-2}.
	The results in Section \ref{sec:sim-results} below  illustrate the three main findings summarized in the main text: (i) a threshold in excess noise amplitude, above which the advanced statistical tests become useful; (ii) an example of simulated environment with results of statistical tests similar to experiment, and a comparison between estimated and actual measure of parametric variability $\Delta P_{\pm}$; (iii) a summary of single fluctuator behavior w.r.t.\ ``fast fluctuator'' and ``slow drift'' statistical test.

	\subsection{Noise simulation procedure}

	Our approach to simulate $1/f$ noise by an ensemble of two-level fluctuators (TLFs) follows the principles reviewed in {\cite{Paladino2014}}.  
	The timeline for the experiment is shown schematically in Figure~\ref{fig:TLF}. The differential error signal $x$ is measured after each burst (i.e., transfer sequence) of length $t_i$,  repeated consecutively $N_i$ times, for a fixed set of burst lengths $i=1 \, \ldots \, L$.

	\begin{figure}[H]
		\tikzset{dots/.style={draw,circle,fill=white,inner sep=1pt,minimum size=4pt}}
		\tikzset{braces/.style={blue,decorate,decoration={brace,amplitude=10pt,mirror},xshift=0.4pt,yshift=-2.4cm}}
		\tikzset{braceText/.style={black,midway,yshift=-10mm, text width=18mm,font=\scriptsize}}
		\tikzset{sepLine/.style={dashed, shorten >=-1cm, shorten <=-1cm}}
		\def\maxNum{20}
		\def\colA{Green!20}
		\def\colB{blue!20}
		\def\yMax{11}          
		\centering            
		\begin{tikzpicture}[scale=1.0] 
		\begin{axis}[% https://tex.stackexchange.com/a/44966
		axis x line=middle,	% center the x axis 
		axis y line=none, % we don't need a y axis line ...
		x label style={at={(axis cs:\maxNum+1.2,0.5)},anchor=south}, 
		xlabel = {\footnotesize$ \mathrm{t} $},
		height=50pt, 
		width=\axisdefaultwidth*1.8, % `height' also changes `width' which is restored here
		xmin=0,
		xmax=\maxNum+2,
		ymin=-\yMax, ymax=\yMax,
		ticks = none ]

		\addplot[\colA,line width=4mm,domain = 0:3.4,opacity=0.9] {x*0};
		\addplot[\colB,line width=4mm,domain = 3.4:6.2,opacity=0.9] {x*0};
		\addplot[\colA,line width=4mm,domain = 6.2:11.3,opacity=0.9] {x*0};
		\addplot[\colB,line width=4mm,domain = 11.3:11.7,opacity=0.9] {x*0};
		\addplot[\colA,line width=4mm,domain = 11.7:13.5,opacity=0.9] {x*0};
		\addplot[\colB,line width=4mm,domain = 13.5:20.5,opacity=0.9] {x*0};
		\pgfplotsinvokeforeach {1,...,\maxNum} {
			\node[dots]  at (axis cs: #1,0){} ;}
		
		\coordinate (p1A) at (axis cs: 1,0) {};
		\coordinate (p1A1) at (axis cs: 2,0) {};
		\coordinate (p1B) at (axis cs: 4,0) {};
		
		\coordinate (p2A) at (axis cs: 5,0) {};
		\coordinate (p2B) at (axis cs: 8,0) {};
		
		\coordinate (p3A) at (axis cs: 9,0) {};
		\coordinate (p3B) at (axis cs: 15,0) {};
		
		\coordinate (p4A) at (axis cs: 16,0) {};
		\coordinate (p4B) at (axis cs: 20,0) {};

		\end{axis}
		
		\draw [braces](p1A) -- (p1B) node[braceText] {$N_1$  bursts of length  $ t_1$};
		\draw [braces](p2A) -- (p2B) node[braceText]  {$N_2$ bursts of length $ t_2$};
		\draw [braces](p3A) -- (p3B) node[braceText]  {$N_3$ bursts of length $ t_3$};
		\draw [braces](p4A) -- (p4B) node[braceText]  {$N_4$ bursts of length $ t_4$}; 
		
		\draw[latex-latex]  ([yshift = 2mm]p1A) -- ([yshift = 2mm]p1A1)       node[black,midway,anchor=south]  {\footnotesize  $\tau_0$};

		\end{tikzpicture}
		\caption{Illustration of the experimental timeline. Circles denote a single instance of the random walk (one burst).
			The colors show the state  of a single fluctuator (green for $0$, blue for $1$) in a particular realization of the simulated disorder.}\label{fig:TLF}
	\end{figure}
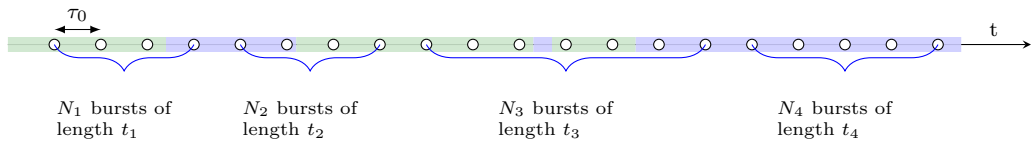

	We explore the regime when the duration of one burst, $t_i \, f^{-1}$, is  much shorter than the detector-limited time interval $\tau_0$ between the repetitions
	(e.g., up to a few $\SI{}{\micro s}$ for the former and on the order of a $\SI{}{\milli s}$ for the latter for device A). Switching events in the environment can only contribute significantly  on time scales of $\tau_0$ and longer, and are neglected during the  short bursts. This means that  $P_{\pm}(\mathrm{t})$ remain fixed during a single instance of the random walk starting at an absolute time $\mathrm{t}$ (an integer multiple of $\tau_0$), and the particular value of $x$ is distributed according to Eq.~\eqref{eq:ptx_model0}   with a particular instance of $P_{\pm}(\mathrm{t})$. In the statistical models defined in Supplementary Note \ref{sec:model-1}--\ref{sec:model-2}, the values of $P_{\pm}$ are drawn from a Dirichlet distribution after the time $\tau_0$ (``fast fluctuator'') or after time $\tau_0 \, N_i$ (``slow drift''). Here, in contrast, we generate $P_{\pm}(\mathrm{t})$ from a continuous-time Markov process characterized by a set of 
	switching rates $\{ \Gamma_1, \Gamma_2 \,\ldots, \Gamma_M \}$ with $\Gamma_m \ge \tau_0^{-1}$.
	
	The value of $P_{\pm}(\mathrm{t})$ is determined by an average over $M$ two-level fluctuators,
	\begin{align} \label{eq:Pfluctuatoraverage}
	P_{\pm}(\mathrm{t}) = \frac{1}{M}\sum_{m=1}^{M} P_{\pm}^{(m)} \lrk{\xi_m(\mathrm{t})} \, ,
	\end{align}
	where $\xi_m(\mathrm{t})=0$ or $1$ is the state of the $m$-th fluctuator at time $\mathrm{t}$.
	Each fluctuator is characterized by two modes, {$P_{\pm}^{(m)}\lrk{0}$ and $P_{\pm}^{(m)}\lrk{1}$}, both drawn once for each full-timeline simulation from a globally fixed Dirichlet distribution. Parameters of the latter are the mean values  $\langle P_{\pm} \rangle$ and the total $\alpha_{\text{noise}}$ which determines the level of excess noise.
	The scatter parameter $\alpha_{\text{noise}}$ of the mode distribution is free; it controls the excess noise level.
	
	Switchings $\xi_m \to 1-\xi_m$ happen randomly with a rate $\Gamma_m$ (same in both directions) independently of other fluctuators. Hence each fluctuator is described by a $2 \times 2$ continuous-time Markov chain transition rate matrix $ 	\lr{\begin{smallmatrix}
		-\Gamma_m  & \Gamma_m\\
		\Gamma_m  & -\Gamma_m
		\end{smallmatrix}} $.
	An example of a timetrace of one fluctuator switching between two modes is shown by colorboxes in Figure~\ref{fig:TLF}.

	Two cases are considered:
	\begin{itemize}
		\item $1/f$-noise: $M=100$,
		$\Gamma_m = (10^{-8} \ldots 10^{0}) \, \tau_0^{-1}$,  with $\log \Gamma_m$ chosen randomly from a uniform distribution.
		%$-8 \leq  \log_{10} (\Gamma_m \tau_0) \leq 0$ with a uniform step.
		\item Single TLF: $M=1$ and $\Gamma_1$ as an adjustable parameter.
	\end{itemize} 
	
	In a simulation, $N_{\text{tot}}$ values  of $x$  are available. For each burst length $t_i$, the corresponding vector $\mathbf{z}_i$ containing the number of  counts for each category of $x$ is compiled in the same way as in experiment. 
	For statistical evaluation of the simulated timetraces,
	we have binned  the simulated counts into five categories, instead of seven as in  Eq.~\eqref{eq:rebinning}.
	
	%The set of switching rates $\{ \Gamma_1, \Gamma_2 \,\ldots, \Gamma_M \}$ determines the noise spectra of $P_{\pm}(\mathrm{t})$.

	%given by with $\Gamma_i \ge \tau_0^{-1}$.

	%$\Gamma_{\text{Martins} } =1- \exp [\Gamma_{\text{slava}} \, \tau_0] \approx \Gamma_{\text{slava}} \tau_0$
	
	We fixed the set of  $L=42$ burst lengths with $1 \leq t_i \leq 100$ and samples sizes
	$\{N_1 \ldots N_{L} \}$ with $N_{\text{tot}}=\sum_{i=1}^{L} N_i = \numlistDef{38022642}$ and 
	$\numlistDef{5.00392e+5} \leq N_i \leq \numlistDef{1.398826e+06} $ to be the same as for experimental results on device A reported in Figure \ref{fig:fig2} in the main text. 
	
	Similarly, parameters of the Dirichlet distribution for the fluctuator modes  {are} chosen to be comparable to the experimental values, 
	$ (\langle\probP\rangle,\langle\probQ\rangle) =(\numlistDef{  2.1e-5; 7.0e-05 })$.

	We quantify the excess noise level by computing directly  the relative standard deviation (RSD) of the simulated time traces $P_{+}(\mathrm{t})$ and
	$P_{-}(\mathrm{t})$, and choosing the maximum:  $\text{RSD}=\max_{s=+,-} {\sigma_s}/{\mu_s}$, where
	\begin{align} \label{eq:mu-sigma}
	\mu_s :=  
	\frac{1 }{   N_{\text{tot}}} \sum_{n=1}^{N_{\text{tot}}}   P_s( n \, \tau_0)  ,
	\quad
	\sigma_s^2 :=
	\frac{1}{   N_{\text{tot}}} \sum_{n=1}^{N_{\text{tot}}}
	\left [ P_s( n \, \tau_0) - \mu_s \right ]^2.
	\end{align} 
	are respectively the  mean and the standard deviation of a particular simulated time trace; 
	the time along the data acquisition time-line $\mathrm{t}= n \, \tau_0$ is measured from $0$.

	\subsection{Noise simulation results\label{sec:sim-results}}

	\subsubsection{Detection of excess noise}
	We have run a number of simulations of $1/f$  noise with varying the noise strength parameter $\alpha_{\text{noise}}$ from  $10^{7.5}$ to $10^{3.6}$ and checked the statistical consistency of the accumulated error counts with the three simple models  discussed  in the main text: baseline model (no excess noise), ``fast fluctuator''  and ``slow drift''.

	\begin{figure}[htpb]
		\centering
		\includegraphics[width=\textwidth]{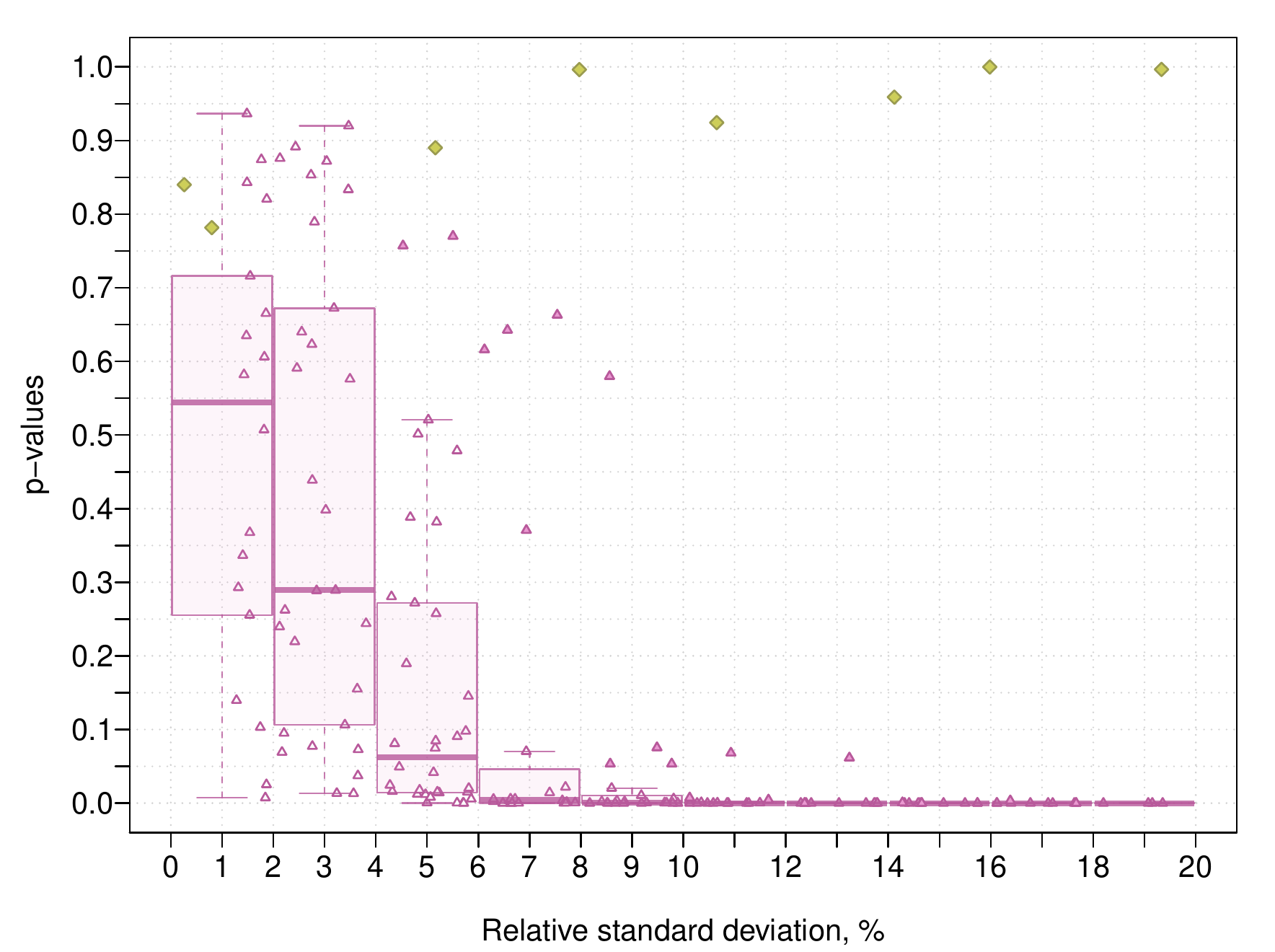}
		\caption[p-value vs RSD]{Results of consistency tests applied to simulated $1/f$ noise for the baseline (triangles) and the ``slow drift'' (diamonds) excess noise models. Horizontal axis: the relative  standard deviation of the simulated timetrace $P_{\pm}(\mathrm{t})$.
			Triangles: combined $p$-values returned by the baseline model. Also depicted the box-and-whisker plots of groups of these values, with filled triangles representing outliers in each group. The diamonds depict the combined $p$-values returned by the ``slow drift'' test.}
		\label{fig:p-vs-noise}
	\end{figure} 
	The results   are plotted in Figure \ref{fig:p-vs-noise}. Each symbol represents a simulation; its abscissa equals the RSD of the respective timetrace and its ordinate is the  {combined} $ p $-value obtained via  consistency testing of the respective  maximum likelihood estimate.  
	
	First we consider the baseline model results shown by triangles in Figure \ref{fig:p-vs-noise}. 
	For very low noise (essentially constant $P_{\pm}(\mathrm{t})$) the null hypothesis is expected to be true, and the $p$-values should be scattered approximately uniformly between $0$ and $1$,  while for stronger noise we expect small $p$ to become progressively more likely. 
	
	To illustrate this transition more clearly, we have added to Figure~\ref{fig:p-vs-noise}  box-and-whisker plots of the same dataset, by binning the baseline $p$-values by the corresponding RSD, i.e., first group with  
	RSD $<  2\%$, second with $2 \% \leq \text{RSD} < 4\%$ and so on, up to $18 \% \leq \text{RSD} < 20\%$.
	The box-and-whisker plots are computed using the standard way, i.e., if $ Q_1 $, $ Q_2 $, $ Q_3 $ are the quartiles of the group and $ \IQR= Q_3-Q_1  $ is the interquartile range, then the lower and upper limits of the  box are placed at $ Q_1 $ and $ Q_3 $, respectively, and the horizontal line within the box represents the median $ Q_2 $. Furthermore, the lower and upper fences are computed as  $ Q_1 - 1.5 \IQR  $ and $ Q_3 + 1.5\IQR $, respectively; if any value in the group is below the lower or above the upper fence, it is considered an outlier {and plotted with a filled symbol}. Finally, whiskers are drawn so that the upper whisker is located either at the upper fence or the maximal value in the group (whichever is smaller); similarly, the lower whisker is located either at the lower fence or the minimal value in the group (whichever is larger).
	
	As it can be observed in the Figure, the first two groups with RSD up to $4\%$ appear to be consistent  with the standard uniform distribution. However, for relative  standard deviation above 4\%  the distribution of the resulting $ p $-values quickly deteriorates and for $ \text{RSD} \geq 6 \% $   already the median of the obtained $ p $-values is well below the {conventional} $ p= 0.05$ threshold.
	{We see a clear evidence of a threshold in environmental noise strengths (quantified by RSD in the simulation timetrace): if the noise is too low, the data are consistent with the baseline model. For larger values of RSD, the $p$-values for the baseline model collapse to very low values, indicating strong rejection of the null hypothesis of no excess noise.}
	
	Next we evaluate to what extent the data collected from $1/f$ noise simulations are consistent with  single-timescale models
	discussed in the main text and  in Supplementary Note \ref{sec:model-1} and \ref{sec:model-2}.  We have found that for this particular type of noise the behavior of the ``fast fluctuator''' test offers only a marginal improvement over the baseline model, hence the corresponding $p$-values are not plotted. 
	
	To the contrary, the ``slow drift'' model is generally consistent with the simulated data even if the excess noise 
	level is high, provided that parameters of the  Dirichlet distribution in the ``slow drift model'' (see Eq.~\ref{eq:Dir_PDF_SD})  are chosen appropriately.
	
	In {Figure \ref{fig:p-vs-noise}}  the diamonds show the combined $p$-values for simulated noise data tested against the ``slow drift'' model using the Monte Carlo estimation procedure of Supplementary Note \ref{sec:model-2-consistency}.
	Here we have used $\max\limits_{s=+,-} \sigma_{s}$ computed from the simulated timetrace $P_{\pm}(\mathrm{t})$ to set \emph{a priori} the concentration parameter $\alpha_{\bullet}$ of the Dirichlet distribution used in the ``slow drift'' test. Despite the value of $\alpha_{\bullet}$  not being optimized (and hence, potentially biasing the test),
	the results indicate that  the ``slow drift'' model for the data is not rejected: 
	the $p$-values are high for the whole range of noise levels considered.

	\subsubsection{Example of statistical methodology applied to the simulated data}
	Here we illustrate in more detail application of statistical tests to a particular simulated timetrace of $1/f$ noise model.
	The simulation has used  $ \alpha_{\text{noise}} = 10^{5.2} $ and
	$ \lra{\probP}=\numlistDef{2.12e-05} $,  $ \lra{\probQ}=\numlistDef{6.9e-05} $.

	The estimated mean $\mu_s$ and the standard deviation $\sigma_s$ of the simulated $P_\pm(\mathrm{t})$ values are $\mu_+=\numlistDef{2.10860698039036e-05}, \mu_- = \numlistDef{6.93810249087584e-05}$ and   
	$\sigma_+=\numlistDef{7.71619720793375e-07}, \sigma_- = \numlistDef{1.36838349724268e-06}$, which gives the  relative standard deviation $\text{RSD}= \numlistDef{3.659381421}\%$. 
	The simulated $P_\pm(\mathrm{t})$ values are well approximated by the normal distribution with the respective parameters; see the histograms of $P_\pm(\mathrm{t})$ in Figure \ref{fig:artificial_data}b.

	Next we describe how the spectrum of the simulated noise was analyzed.
	Let $ R= 2^{18} $, $ K = \left\lfloor\frac{N_{\text{tot}}}{R}\right\rfloor=145  $.
	Given the signal $ P_s( n \, \tau_0) $, $ n=1,2,\ldots, N_{\text{tot}} $, where $ s \in \{-, +\} $, its power spectral density (PSD) is estimated by splitting the signal into $ K $ non-overlapping segments of length $ R $ (the signal components with $ n> KR $ are discarded) and computing   the periodogram  for each segment and averaging the results for the $ K $ segments.
	Each periodogram is found by computing the squared magnitude of the discrete Fourier transform and dividing the result by $ R $; moreover, the values at all frequencies except 0 and  0.5 are doubled since the one-sided periodogram is considered.

	In Figure \ref{fig:artificial_data}a we show the PSD estimates of the simulated timetrace which show the signature $1/f$ roll-off.
	
	\begin{figure}[ht]
		\centering
		\includegraphics[width=\textwidth]{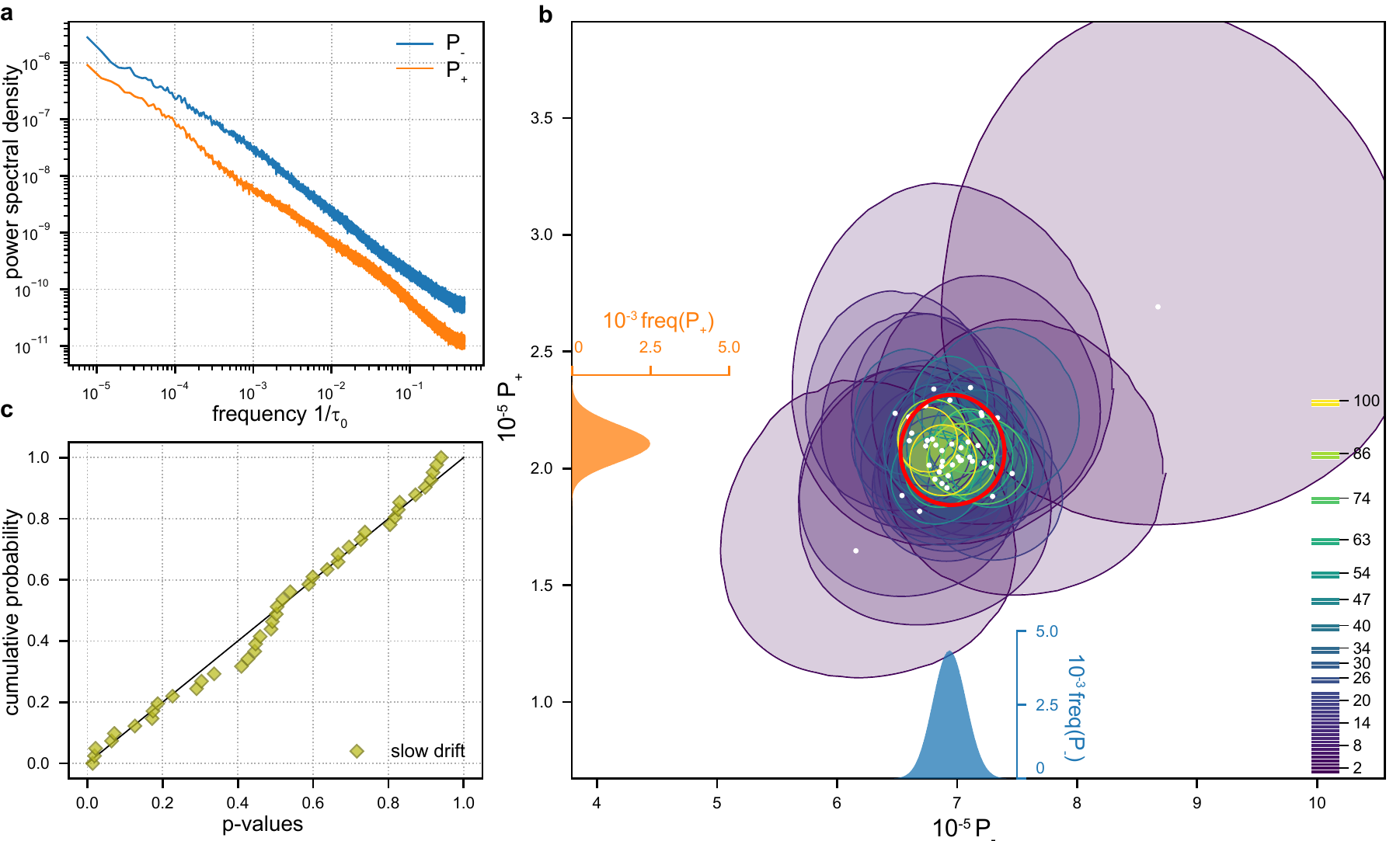}
		\caption{
			(a) The PSD estimates of the simulated timetrace.
			(b) White dots: maximum-likelihood estimates of $\probPQ$ in the baseline model for each sequence length separately. Also $p>0.05$ consistency regions are shown (color indicating the sequence length).  The   histograms of the simulated $P_\pm(\mathrm{t})$ values are depicted on the axes of the plot.
			(c) Empirical cumulative distribution of  $ p $-values for the slow drift model in comparison to the uniform distribution (black line).
		}
		\label{fig:artificial_data}
	\end{figure}

	Next, we analyze the simulated sample in the same way as the experimental data for device A,
	c.f.\ Figure \ref{fig:fig2} in the main text.

	For the baseline model, the maximum likelihood gives the estimate  $ \lr{\probP,\probQ}=\lr{\numlistDef{2.07161e-05;6.95532e-05} }$,
	which correlates well with the underlying $\mu_{\pm} \pm \sigma_{\pm}$.  Next, we employ Fisherian significance tests separately for each block of $N_i$ counts for a fixed burst length $t_i$ to define  consistency regions of $p$-value greater than $0.05$ in the parameter space $(P_{+}, P_{-})$ where the baseline model cannot be rejected at this significance level. These quasielliptic consistency regions are depicted in Figure~\ref{fig:artificial_data}b; again, their overlap is only partial. Fisher's meta-analysis method  yields the combined $p$-value $ \numlistDef{0.03914580} $, thus the baseline model is nominally rejected.
	
	The simulated data then were tested against the ``slow drift'' model with the optimized (as explained in Supplementary Note~\ref{sec:model-2-consistency}) parameters $ \boldsymbol{\alpha}^* = (\numlistDef{5.00782997e+03; 7.19934995e+07; 1.49155999e+03}) $. The cumulative distribution of the obtained $ p $-values for each sequence length  is shown in Figure~\ref{fig:artificial_data}c;
	it minimizes the empirical distance measure \eqref{eq:distanceMeasure} and is close to being uniform.
	Quantitatively,
	the combined $ p $-value  equals $ \numlistDef{0.4471313} $,  implying that the simulated data are not inconsistent with the slow drift model.
	
	The standard deviations of $\probPQ$ in the Dirichlet distribution $\Dir(\boldsymbol{\alpha}^*)$ are $\numlistDef{5.363933e-07}$ and $\numlistDef{9.828273e-07}$, respectively, producing the {$1\sigma$} estimates
	\[%	P+ : 5.363933e-07  P- : 9.828273e-07	(2.07+/-0.054)	(6.96+/-0.098)
	{ P}_+  = (2.07 \pm 0.05) \times 10^{-5},
	\quad
	{ P}_-  = (6.96 \pm 0.10) \times 10^{-5} \quad (\text{``slow drift'' estimate})
	\]
	which compare well with  $\mu_{\pm} \pm \sigma_{\pm}$ accessible for  the simulated  noise,
	\[ %(2.11+/-0.077) (6.94+/-0.14)
	\probP =(2.11 \pm 0.08) \times 10^{-5}
	, \quad
	\probQ= (6.94 \pm 0.14) \times 10^{-5} \quad (\text{simulated noise})
	\]

	\subsubsection{Simulations of a single fluctuator}	
	We have explored the excess noise model with a single fluctuator ($M=1$) with the switching rate 
	$\Gamma_1$ varied  in simulations from 
	$ 10^{-6} \tau_0^{-1}$ to 
	$\tau_0^{-1}$. The distribution of
	$P_{\pm}(\mathrm{t})$ in this case is bimodal, with  the two modes,
	$P_{\pm}^{(1)}[0]$ and  $P_{\pm}^{(1)}[1]$, chosen randomly before the start of the simulation.
	The corresponding noise power spectral density is a Lorentzian that crosses over from a constant to $1/f^2$ at frequency $f \sim \Gamma_1$.   
	
	With respect to the baseline and the ``slow drift'' statistical tests, the single-fluctuator model behaves similarly to the $1/f$ case considered above, regardless of the choice of $\Gamma_1$:
	a sufficiently large distance between the two modes generates counts inconsistent with the baseline, but compatible with the ``slow drift'' model  with an appropriately chosen $\Delta P_{\pm}$. 
	
	The results of ``fast flcutuator'' tests do, in general, depend on the switching speed $\Gamma_1$: for fast $\Gamma_1 \sim \tau_0^{-1}$ the simulated data typically are consistent with the ``fast fluctuator''  model, while for slower $\Gamma_1 \ll \tau_0^{-1}$ the  simulated  data reject  the ``fast fluctuator'' noise model.
	
	These observations can be understood considering the characteristic  frequencies to which our statistical tests are sensitive: ``slow drift''  probes the low-frequency part of the noise spectrum (switching rates on the order of $1/( N_i \tau_0) \sim 10^{-5} \tau_0^{-1}$) while the ``fast fluctuator'' tests  relatively high frequencies ($\sim \tau_0^{-1}$), at which the noise from a slow ($\Gamma_1 \ll \tau^{-1}$) single fluctuator is sufficiently suppressed by the  $1/f^2$ roll-off.

	\section{Slow drift model: variance of random walker's position}\label{sec:deltaX2} 
	In this section we derive the formula for the variance of the random walker's position $\Delta x^2$ used in the concluding part of the main text.
	
	Consider repeatedly sampling random walker's position in the slow drift model; we are interested in the sample variance. However, the  samples are correlated, since the model assumes using the same parameter $ \boldsymbol \theta \sim \Dir(\boldsymbol \alpha) $ for several ($ N $)  successive random walks.  Henceforth, we consider the following scenario: 
	draw $ \boldsymbol{\theta}^{(1)}\sim \Dir (\boldsymbol\alpha) $ and run $ t  $ steps of  the random walk with step probabilities $ \boldsymbol{\theta}^{(1)} $ for $ N $ times;   let $ X_1^{(1)} $, $ X_2^{(1)} $, $\ldots$, $ X_N^{(1)} $ be the random walker's position after the respective random walk has been completed.  Afterwards, the step probabilities reset, i.e., a new parameter $ \boldsymbol \theta^{(2)} \sim \Dir(\boldsymbol \alpha) $  is independently drawn, and $ N $ times random walk of $ t $ steps is run with step probabilities $ \boldsymbol{\theta} $, resulting in  random variables $ X_1^{(2)} $, $ X_2^{(2)} $, $\ldots$, $ X_N^{(2)} $. The process is continued until, say, $ K $ blocks  $ {\mb X^{(k)}} := \lr{X^{(k)}_1,X^{(k)}_2,\ldots,X^{(k)}_N} $, $ k=1,2,\ldots,K $, are obtained.  It is important to stress that
	\begin{enumerate}
		\item the blocks $ {\mb X^{(1)}} ,  {\mb X^{(2)}} , \ldots,  {\mb X^{(K)}}   $ are assumed to be pairwise independent;
		\item within each block, variables $ X_i^{(k)} $ and $ X_j^{(k)} $, $ i\neq j $, $ k =1,2,\ldots, K $, are assumed to be conditionally independent given $ \boldsymbol \theta^{(k)} $, i.e., 
		\[ 
		(X_i^{(k)}  \perp \!\!\!\perp X_j^{(k)}) \mid \boldsymbol \theta^{(k)}.
		\]
	\end{enumerate}
	Let $ M:=KN $; we are interested in the quantity
	\[ 
	S := \frac{1}{M} \sum_{k=1}^{K} \sum_{i=1}^{N}  \lr{X^{(k)}_i   - \barX }^2,
	\]
	where $ \barX :=  \frac{1}{M}\sum_{k=1}^{K} \sum_{i=1}^{N} X_i^{(k)} $. In particular, the  task is to find $ \Delta x^2 :=  \Expect(S) $ in the $ K \to \infty $ limit.
	
	\begin{claim}
		The expectation  $ \Delta x^2 =  \Expect(S) $ in the $P_{\pm}\ll 1$, $ K \to \infty $ limit   satisfies
		\[ 
		\Delta x^2  = 
		\frac{A }{1+\alpha_\bullet}\lr{\frac{ t\alpha_ \bullet(M-1) }{M}   + \frac{t^2 (K-1)}{K} }
		\approx (\langle \probP \rangle +\langle \probPQ \rangle) \, t + (\Delta\probP ^2 +\Delta \probQ^2) \, t^2 ,
		\]
		where
		\begin{align*}
		& A :=  \lr{\tilde {\alpha}_{-1}+\tilde {\alpha}_{1}  } - \lr{\tilde {\alpha}_{-1}-\tilde {\alpha}_{1}}^2   ;\\ 
		& \tilde \alpha_{\pm 1} :=  \frac{\alpha_{\pm 1}}{\alpha_\bullet};  \quad \alpha_\bullet := \alpha_{-1}+\alpha_0+\alpha_1;\\
		&  \langle \probPQ \rangle := \Expect(\probPQ ) =   \tilde \alpha_{\pm 1};\\
		& \Delta\probPQ ^2 := \Var(\probPQ ) = \frac{\tilde \alpha_{\pm 1} (1- \tilde \alpha_{\pm 1})}{1+ \alpha_\bullet  }.
		\end{align*}
	\end{claim}
	\begin{proof}
		We start by noting that each block $ \mb X^{(k)} $ is an independent observation of the random variable $ \mb X = (X_1, X_2, \ldots, X_N) $, where $ \mb  X$ is obtained by the process above with $ K=1 $. Let us define a random variable  $ 	Y=  \frac{1}{N} (X_1 + \ldots + X_N) $,  
		then  there  are $ K $   iid variables $ Y_1, \ldots, Y_K \sim Y $ corresponding to the blocks $ \mb X^{(1)} $, \ldots, $ \mb X^{(K)} $ . We can express
		\[ 
		S= 
		\frac{1}{M} \sum_{i,k}  \lr{X^{(k)}_i   }^2   - \frac{2}{M}  \sum_{i,k}  \lr{X^{(k)}_i   } \barX  + \barX^2
		=
		\frac{1}{M} \sum_{i,k}  \lr{X^{(k)}_i  }^2 - \barX ^2
		\]
		and
		\[ 
		\barX = \frac{1}{K} (Y_1 + Y_2 + \ldots + Y_K).
		\]

		Since $ X^{(k)}_i $ are identically distributed, we have
		\[ 
		\Expect\lr{\frac{1}{M}\sum_{k=1}^{K} \sum_{i=1}^{N}  \lr{X^{(k)}_i  }^2}
		=
		\frac{M}{M} \cdot \Expect(X_1^2) =  \Var(X_1) + \Expect(X_1)^2;
		\]
		since $ Y_1, \ldots, Y_K $ are iid, we have
		\[ 
		(Y_1 + Y_2 + \ldots + Y_K) ^2 = \sum_{k=1}^K Y_k^2  +  \sum_{k=1}^K \sum_{\substack{l=1\\ l \neq k}}^K Y_k Y_l
		\]
		and
		\[ 
		\Expect \lr{ \barX^2}  = \frac{1}{K} \Expect(Y^2)  +\frac{K-1}{K} \Expect(Y)^2 =  \frac{1}{K}   \Var(Y)  + \Expect(Y)^2 .
		\]
		Consequently,
		\[ 
		\Delta x^2 =   \Var(X_1) + \Expect(X_1)^2 -  \frac{1}{K}   \Var(Y)  - \Expect(Y)^2 .
		\]
		Let us show that 
		\begin{align}
		&  \Expect(X_1) = t \lr{\tilde \alpha_1  -  \tilde \alpha_{-1}}  \label{eq:Xi_expect} \\
		& \Var(X_1) = \frac{A t\alpha_\bullet  }{1+\alpha_\bullet} +\frac{A t^2  }{1+\alpha_\bullet}
		,  \label{eq:Xi_var} \\
		& \Expect(Y) = \Expect(X_1) , \label{eq:Y_expect} \\
		&   \Var(Y)= \frac{A t\alpha_\bullet  }{N(1+\alpha_\bullet)} + \frac{A t^2  }{1+\alpha_\bullet} 
		\label{eq:varY},
		\end{align}
		then we will arrive at 
		\[ 
		\Delta x^2 
		=   \Var(X_1)  -  \frac{1}{K}   \Var(Y) 
		=
		\frac{A t\alpha_\bullet (M-1) }{M(1+\alpha_\bullet)}  +  \frac{A t^2 (K-1) }{K(1+\alpha_\bullet)},
		\]
		as desired.

		\textbf{Equalities \eqref{eq:Xi_expect} and \eqref{eq:Xi_var}.} 
		If $\boldsymbol \theta = (\probQ,\probO,\probP)   $ is a fixed parameter, then
		the position of the random walker after $ t $ steps with step probabilities given by  $ \boldsymbol \theta $ is given by
		$   \delta_1 + \ldots + \delta_t $, where $ \delta_i \sim \delta $ are iid and
		\[ 
		\delta = 
		\begin{cases}\
		-1, & \text{ with prob. }  \probQ, \\
		0, & \text{ with prob. } \probO,    \\
		+1, & \text{ with prob. } \probP.
		\end{cases}
		\]
		Consequently, for the Dirichlet-distributed $ \boldsymbol\theta $ the conditional expectation / variance $ \Expect (X_1 \mid  \boldsymbol\theta)  $ and $ \Var (X_1 \mid \boldsymbol \theta)   $
		satisfy
		\begin{align*}
		& \Expect (X_1 \mid  \boldsymbol\theta) =     t \lr{\probP  -\probQ}\\
		& \Var (X_1 \mid \boldsymbol \theta)  = t  \lr{\probP  + \probQ }- t \lr{\probP  -  \probQ}^2.
		\end{align*}
		By the laws of total expectation / variance, we obtain  \eqref{eq:Xi_expect} and \eqref{eq:Xi_var}:
		\begin{align*}
		\Expect (X_1)   & = \Expect \lr {\Expect (X_1 \mid  \boldsymbol\theta)} = t(\tilde \alpha_1  -  \tilde \alpha_{-1})\\
		\Var(X_1)  &  = \Expect(\Var(X_1 \mid \boldsymbol \theta)) + \Var (\Expect (X_1 \mid \boldsymbol \theta)) \\
		& = t \Expect\lr{ \probP  + \probQ   + 2 \probP \probQ   -  \probP^2 - \probQ^2  }  + t^2 \Var \lr{ \probP - \probQ } \\
		&  =  \frac{t \alpha_\bullet A}{1+\alpha_\bullet}  + \frac{t^2A}{1+\alpha_\bullet}.
		\end{align*} 
		To show the last equality, recall the relevant properties of   $ (\probQ,\probO,\probP)  \sim \Dir(\boldsymbol\alpha)  $:
		\begin{align*}
		& \Expect(\probPQ) = {\tilde \alpha} _{\pm 1};\\
		& \Var(\probPQ ) = \frac{\tilde \alpha_{\pm 1} (1- \tilde \alpha_{\pm 1})}{1+ \alpha_\bullet  }; \\
		&  \Cov(\probP ; \probQ)  = \frac{-\tilde \alpha_{-1}  \tilde \alpha_1}{1+ \alpha_\bullet}.
		\end{align*}
		Thus
		\begin{align*}
		& \Var \lr{ \probP - \probQ } = \Var \lr{ \probP  } +  \Var \lr{ \probQ  }  - 2 \Cov(\probP ; \probQ)   = \frac{A}{1+\alpha_\bullet}
		\end{align*}
		and
		\begin{align*}
		& \Expect\lr{ \probP  + \probQ   + 2 \probP \probQ   -  \probP^2 - \probQ^2  }  \\
		&=
		\Expect\lr{\probP  + \probQ}  +2\lr{ \Cov(\probP ; \probQ)  +  \Expect(\probP) \Expect(  \probQ) }-  \lr { \Var \lr{ \probP  } + \Expect(\probP)^2}-  \lr{\Var \lr{ \probQ  }   + \Expect(\probQ)^2}\\
		&=
		\Expect\lr{\probP  + \probQ} -\lr{ \Expect\lr{ \probP - \probQ }}^2
		- \Var \lr{ \probP - \probQ } 
		\\
		& = A
		- \frac{A}{1+\alpha_\bullet} 
		= \frac{\alpha_\bullet A}{1+\alpha_\bullet}. 
		\end{align*}

		\textbf{Equality \eqref{eq:Y_expect}.}  This equality trivially follows from the definition of $ Y $ and the fact that $ X_i $ are identically distributed:
		\[ 
		\Expect(Y) = \frac{1}{N} \Expect(X_1 + \ldots +X_N) = \frac{N}{N} \Expect(X_1).
		\]

		\textbf{Equality \eqref{eq:varY}.}   The variance of  the sum $ X_1 + \ldots + X_N $ is the sum of the covariances:
		\[ 
		\Var(X_1 + \ldots + X_N) = \sum_{i=1}^N \Var(X_i) + 2 \sum_{1 \leq i <  j <N}\Cov (X_i, X_j)
		\]
		However, $ X_i $ are identically distributed, therefore
		\begin{equation}\label{eq:varY-1}
		\Var (Y) = \frac{1}{N^2} \lr{  N \Var (X_1) +  (N^2-N) \Cov (X_1, X_2) }
		\end{equation}
		$ \Var(X_1) $ is given by \eqref{eq:Xi_var}, it remains to find $ \Cov (X_1, X_2)  $. By the {law of total covariance},
		\[ 
		\Cov(X_1, X_2) = 
		\Expect\lr{ \Cov\lr{X_1, X_2 \mid \boldsymbol \theta }  }
		+ 
		\Cov\lr{ \Expect\lr{ X_1 \mid \boldsymbol \theta } , \Expect\lr{ X_2 \mid \boldsymbol \theta }  }
		\]
		Since $ X_1, X_2 $ are conditionally independent given $ \boldsymbol \theta $, the conditional covariance vanishes: $  \Cov\lr{X_1, X_2 \mid \boldsymbol \theta } =0 $. Moreover, as $ X_1, X_2 $ are identically distributed,
		\[ 
		\Expect\lr{ X_1 \mid \boldsymbol \theta } = \Expect\lr{ X_2 \mid \boldsymbol \theta } = t \lr{\probP - \probQ  },
		\]
		thus
		\[ 
		\Cov(X_1, X_2) = t^2 \Var \lr{\probP - \probQ  } =\frac{t^2 A }{1+\alpha_\bullet}.
		\]
		We arrive at
		\begin{equation*} %\label{eq:varY-2}
		\Var(Y)=
		\frac{1}{N} \lr{
			A \cdot\frac{t\alpha_\bullet + t^2}{1+\alpha_\bullet} + (N-1)\frac{t^2 A }{1+\alpha_\bullet}
		}
		=
		\frac{A t\alpha_\bullet  }{N(1+\alpha_\bullet)} + \frac{A t^2  }{1+\alpha_\bullet} ,
		\end{equation*}
		which concludes the proof.
	\end{proof}
	
	\section{Proof for spread condition}\label{sec:proof-spread-condition}
	
	This section contains the proof for the spread condition shown in \eqref{eq:sc} in the main text (restated here for convenience).
	\begin{claim}
		For any random process on the line with steps of length 1, the spread condition
		\begin{equation}\label{eq:sc2}
		\sum\limits_{y=-\infty}^{x-1}p_y^t \leq \sum\limits_{y=-\infty}^{x}p_y^{t+1} \leq \sum\limits_{y=-\infty}^{x+1}p_y^{t} \quad \text{for all } x,
		\end{equation}
		must be satisfied.
	\end{claim}
	
	\begin{proof}
		If the random process is at location $y \leq x-1$ after $t$ time steps, it must be at a location $y'\leq x$ after $t+1$ time steps. Hence, $\sum_{y=-\infty}^{x-1}p_y^t\leq\sum_{y=-\infty}^{x}p_y^{t+1}$. 
		
		On the other hand, if the process is at a location $y\leq x$ after $t+1$ time steps, it has been at a location $y'\leq x+1$ after $t$ time steps. Hence, $\sum_{y=-\infty}^{x}p_y^{t+1}\leq\sum_{y=-\infty}^{x+1}p_y^{t}$.
	\end{proof}
	We note that the claim applies not just to Markov processes but to any random process that can move at most distance 1 in one time step. For example, it applies to processes where the transition probabilities depend not just on the current location but also on locations in previous time steps.
	
	\begin{claim}\label{thm:spread_prt2}
		If two probability distributions $(p_x^t)_{x\in \mbb Z}$ and $(p_x^{t+1})_{x\in \mbb Z}$ satisfy
		the spread condition \eqref{eq:sc2},
		then there exists a set of transition probabilities $P_{\pm 1}^{(x,t)}$ such that $(p_x^{t+1})_{x\in \mbb Z}$ is generated from $(p_x^t)_{x\in \mbb Z}$ by a Markov process
		\[
		p_{x}^{t+1}= 
		P_{+1}^{(x-1,t)} p_{x-1}^t 
		+
		\lr {1   - P_{+1}^{(x,t)}-P_{-1}^{(x,t)}} p_{x}^t 
		+
		P_{-1}^{(x+1,t)} p_{x+1}^t , 
		\quad
		x \in \mbb Z.
		\]
		
		%For any {collection of distributions} $(p_x^t)_{x \in \mbb Z}$ that satisfy the spread condition \eqref{eq:sc2}, there is a Markov process on the line with steps of length 1 that produces the distribution $(p_x^t)_{x \in \mbb Z}$ in $t$ steps.
	\end{claim}
	\begin{proof}
		Let $q_x^t=\sum_{y=-\infty}^{x}p_y^{t}$; then the spread condition is equivalent to
		\[
		q^t_{x-1} \leq q^{t+1}_{x} \leq q^t_{x+1}, 
		\quad \text{for all }x \in \mbb Z.
		\]
		Consider a Markov process in which the probability of moving left from a location $x$ at time $t$ is defined by
		\begin{align*}
		P_{-1}^{(x,t)}=
		\begin{cases}
		0 &\text{if }q_{x-1}^t \geq q_{x-1}^{t+1}\\
		\frac 1{p_x^t}\lr{{q_{x-1}^{t+1}-q_{x-1}^t}} & \text{otherwise},
		\end{cases}
		\end{align*}
		and the probability of moving right is defined by
		\begin{align*}
		P_{+1}^{(x,t)}= %r_{m,right}^t=
		\begin{cases}
		0 
		&\text{if }q_{x}^t \leq q_{x}^{t+1}\\
		\frac 1 {p_{x}^{t}}  \lr{{q_{x}^{t}-q_{x}^{t+1}}}
		&\text{otherwise}.
		\end{cases}
		\end{align*}
		The probability to stay at $x$ is defined as $P_{0}^{(x,t)}:=1-P_{+1}^{(x,t)}-P_{-1}^{(x,t)}$.
		Notice that the spread condition implies
		\[
		q_{x-1}^{t+1} - q_{x-1}^t \leq q_{x}^t -q_{x-1}^t =  p_x^t
		\quad \text{and} \quad
		q_{x}^t - q_{x}^{t+1} =q_{x-1}^{t}+p^t_x  -q_{x}^{t+1}   \leq  p^t_x,
		\]
		thus $P_{-1}^{(x,t)}\leq 1$ and $P_{+1}^{(x,t)} \leq 1$. 
		Finally, we have $P_{0}^{(x,t)}\geq 0$. To see that, it suffices to consider the case when $P_{+1}^{(x,t)}$ and $P_{-1}^{(x,t)}$ are both positive; then we have to show 
		\[
		\lr{{q_{x}^{t}-q_{x}^{t+1}}}
		+
		\lr{{q_{x-1}^{t+1}-q_{x-1}^t}} \leq p^t_x.
		\]
		However, this inequality clearly holds, since
		\[
		\lr{{q_{x}^{t}-q_{x}^{t+1}}}
		+
		\lr{{q_{x-1}^{t+1}-q_{x-1}^t}} 
		=
		\lr{{q_{x}^{t}-q_{x-1}^t}}
		+
		\lr{{q_{x-1}^{t+1}-q_{x}^{t+1}}} 
		=p^t_x - p^{t+1}_x
		\leq p^t_x.
		\]
		Therefore we have defined valid transition probabilities.

		To see that this process produces $(p_x^{t+1})_{x \in \mbb Z}$, let $\overline{q}_x^{t+1}$ be the probability of being at a location $x'\leq x$ after applying these transition probabilities to the distribution $(p_x^{t})_{x \in \mbb Z}$. We show that $\overline{q}_x^{t+1}=q_x^{t+1}$ for all $x$, thus the probability of being at any particular $x_0$ equals $q_{x_0}^{t+1}-q_{x_0-1}^{t+1}=p_{x_0}^{t+1}$.
		We consider two cases. 
		\begin{enumerate}
			\item If $q_x^{t}\leq q_x^{t+1}$, we have $P_{+1}^{(x,t)}=0$ and
			\begin{align*}
			\overline{q}_x^{t+1}
			=
			q_x^{t}+p_{x+1}^{t}P_{-1}^{(x+1,t)}
			=
			q_{x}^{t}+\left(q_{x}^{t+1}-q_{x}^{t} \right)
			=
			q^{t+1}_{x} .
			\end{align*}
			\item  
			If $q_{x}^{t}>q_{x}^{t+1}$, we have $P_{-1}^{(x+1,t)}=0 $ %r_{x+1,left}^{t}=0
			and 
			\begin{align*}
			\overline{q}_{x}^{t+1}
			=
			q_{x-1}^{t}+p_{x}^{t}
			\left(1 - P_{+1}^{(x,t)}\right)
			=
			q_{x-1}^{t}+p_{x}^{t}-\left( q_{x}^{t}-q_{x}^{t+1} \right)
			=
			q_{x}^{t+1} , 
			\end{align*}
		\end{enumerate}
		which concludes the proof.
	\end{proof}
	A consequence of these two claims is that, given just the probabilities $p_x^t$, we cannot distinguish whether they come from a (possibly non-stationary, not translation-invariant) Markov process or from a more general process that moves at most distance 1 in one time step. (In the second case, the spread condition will be satisfied and then, because of Claim \ref{thm:spread_prt2}, there will be a time and location dependent Markov process that gives the same $p_x^t$.)

\end{document}